%% file: main.tex
\documentclass[11pt]{article}
\usepackage{amsmath}
\usepackage{amsthm}
\usepackage{amssymb}
\usepackage{algorithm}
\usepackage{subfig}
\usepackage{color}
\usepackage[english]{babel}
\usepackage{graphicx}
\usepackage{wrapfig,epsfig}
\usepackage{epstopdf}
\usepackage{url}
\usepackage{graphicx}
\usepackage{color}
\usepackage{epstopdf}
\usepackage{algpseudocode}
\usepackage{scrextend}
\usepackage[T1]{fontenc}
\usepackage{bbm}
\usepackage{comment}
\usepackage{authblk}
\usepackage{multicol}
\usepackage{dsfont}

\usepackage{tikz}
\usepackage{hyperref}  
\hypersetup{colorlinks=true,citecolor=blue,linkcolor=blue} 
\usetikzlibrary{arrows}
\usepackage[margin=1in]{geometry}
\graphicspath{{./figs/}}

\author{
  Vasileios Nakos\thanks{\texttt{vasileiosnakos@g.harvard.edu} Harvard University.}
  \quad
  Zhao Song\thanks{\texttt{zhaos@g.harvard.edu} Harvard University \& UT-Austin. Work done while visiting Harvard University, hosted by Jelani Nelson.}
}  

\date{}

\title{Stronger $\ell_2/\ell_2$ Compressed Sensing; Without Iterating}

\newtheorem{theorem}{Theorem}[section]
\newtheorem{lemma}[theorem]{Lemma}
\newtheorem{definition}[theorem]{Definition}

\newtheorem{fact}[theorem]{Fact}
\newtheorem{remark}[theorem]{Remark}
\newtheorem{claim}[theorem]{Claim}

\newtheorem{problem}[theorem]{Problem}

\newcommand{\wh}{\widehat}

\newcommand{\ov}{\overline}
\newcommand{\eps}{\epsilon}
\newcommand{\N}{\mathcal{N}}
\newcommand{\R}{\mathbb{R}}

\renewcommand{\varepsilon}{\epsilon}

\renewcommand{\hat}{\wh}
\renewcommand{\bar}{\ov}
\renewcommand{\eps}{\epsilon}

\DeclareMathOperator*{\E}{{\bf {E}}}

\DeclareMathOperator*{\median}{median}

\DeclareMathOperator{\sparse}{sparse}
\DeclareMathOperator{\poly}{poly}

\newcommand{\Zhao}[1]{{\color{red}[Zhao: #1]}}

\makeatletter
\newcommand*{\RN}[1]{\expandafter\@slowromancap\romannumeral #1@}
\makeatother

\usepackage{lineno}

\input{libtheorems.tex}

\begin{document}

\begin{titlepage}
  \maketitle
  \begin{abstract}
\input{abstract}

  \end{abstract}
  \thispagestyle{empty}
\end{titlepage}

\newpage



\input{intro}

\input{tech}






\newpage
{\hypersetup{linkcolor=black}
\tableofcontents
}
\newpage
\appendix
\input{btree}
\input{pruning}
\input{tail}
\input{combine}
\input{ack}

\newpage
\addcontentsline{toc}{section}{References}
\bibliographystyle{alpha}
\bibliography{ref}
\newpage





\end{document}

%% file: libtheorems.tex
\usepackage{etoolbox}

\makeatletter

\newcommand{\define}[4][ignore]{%
  \ifstrequal{#1}{ignore}{}{
  \@namedef{thmtitle@#2}{#1}}%
  \@namedef{thm@#2}{#4}%
  \@namedef{thmtypen@#2}{lemma}%
  \newtheorem{thmtype@#2}[theorem]{#3}%
  \newtheorem*{thmtypealt@#2}{#3~\ref{#2}}%
}

\newcommand{\state}[1]{%
  \@namedef{curthm}{#1}
  \@ifundefined{thmtitle@#1}{
  \begin{thmtype@#1}
    }{
  \begin{thmtype@#1}[\@nameuse{thmtitle@#1}]
  }
    \label{#1}
    \@nameuse{thm@#1}
  \end{thmtype@#1}
  \@ifundefined{thmdone@#1}{
  \@namedef{thmdone@#1}{stated}%
  }{}
}

\newcommand{\restate}[1]{%
  \@namedef{curthm}{#1}
  \@ifundefined{thmtitle@#1}{
    \begin{thmtypealt@#1}
    }{
  \begin{thmtypealt@#1}[\@nameuse{thmtitle@#1}]
  }
    \@nameuse{thm@#1}
  \end{thmtypealt@#1}
  \@ifundefined{thmdone@#1}{
  \@namedef{thmdone@#1}{stated}%
  }{}
}

\newcommand{\thmlabel}[1]{
  \@ifundefined{thmdone@\@nameuse{curthm}}{\label{#1}
    }{\tag*{\eqref{#1}}}
}
\makeatother

%% file: abstract.tex
We consider the extensively studied problem of $\ell_2/\ell_2$ compressed sensing. The main contribution of our work is an improvement over [Gilbert, Li, Porat and Strauss, STOC 2010] with faster decoding time and significantly smaller column sparsity, answering two open questions of the aforementioned work.

Previous work on sublinear-time compressed sensing employed an iterative procedure, recovering the heavy coordinates in phases. We completely depart from that framework, and give the first sublinear-time $\ell_2/\ell_2$ scheme which achieves the optimal number of measurements without iterating; this new approach is the key step to our progress. Towards that, we satisfy the $\ell_2/\ell_2$ guarantee by exploiting the heaviness of coordinates in a way that was not exploited in previous work. Via our techniques we obtain improved results for various sparse recovery tasks, and indicate possible further applications to problems in the field, to which the aforementioned iterative procedure creates significant obstructions.

%% file: intro.tex
\section{Introduction}

Compressed Sensing, or sparse recovery, is a powerful mathematical framework the goal of which is to reconstruct an approximately $k$-sparse vector $x\in \mathbb{R}^n$ from linear measurements $y = \Phi x$, where $\Phi \in \mathbb{R}^{m \times n}$. The most important goal is to reduce the number of measurements $m$ needed to approximate the vector $x$, avoiding the linear dependence on $n$. In discrete signal processing, where this framework was initiated \cite{crt06,d06}, the core principle that the sparsity of a signal can be exploited to recover it using much fewer samples than the Shannon-Nyquist Theorem. We refer to the matrix $\Phi$ as the sketching or sensing matrix, and $y = \Phi x$ as the sketch of vector $x$.

Sparse recovery is the primary task of interest in a number of applications, such as image processing \cite{tlwdbskb06,ldp07,ddtlskb08}, design pooling schemes for biological tests \cite{ecgrnomg09,dwgn+13}, pattern matching \cite{cepr07}, combinatorial group testing \cite{saz09, esaz09,kbg10}, localizing sources in sensor networks \cite{zbsgb05,zpb06}, as well as neuroscience \cite{gs12}. Furthermore, not surprisingly, tracking heavy hitters in data streams, also known as frequent items, can be captured by the sparse recovery framework \cite{m05,ch09,kszc03,i07}. In practice, streaming algorithms for detecting heavy hitters have been used to find popular destination addresses and heavy bandwidth users by AT\&T \cite{cjkmss04} or answer ``iceberg queries'' in databases \cite{fsgmu99}.


%

Sparse recovery attracts researchers from different communities, from both theoretical and practical perspective. During the last ten years, hundreds of papers have been published by theoretical computer scientists, applied mathematicians and electrical engineers that specialize in compressed sensing. While numerous algorithms using space linear in the universe size $n$ are known, \cite{d06,crt06,ir08,nt08,bir08,bd09,sv16} to name a few, our goal is to obtain algorithms that are sublinear, something that is crucial in many applications.

The desirable quantities we want to optimize may vary depending on the application. For example, in network management, $x_i$ could denote the total number of packets with destination $i$ passing through a network router. In such an application, storing the sketching matrix explicitly is typically not a tenable solution, since this would lead to an enormous space consumption; the number of possible IP addresses is $2^{32}$. Moreover, both the query and the update time should be very fast, in order to avoid congestion on the network. Incremental updates to $x$ come rapidly, and the changes to the sketch should also be implemented very fast; we note that in this case, even poly-logarithmic factors might be prohibitive. Interested readers can refer to \cite{kszc03, ev03} for more information about streaming algorithms for network management applications.

{\emph{``The goal of that research is to obtain encoding and recovery schemes with good compression rate (i.e., short sketch lengths) as well as good algorithmic properties (i.e., low encoding, update and recovery times).'' $\quad \quad$ -- Anna Gilbert and Piotr Indyk} \cite{gi10}} 

Sparse recovery schemes that are optimal across all axis are a challenge and an important theoretical and practical problem. For most sparse recovery tasks, we have algorithms that achieve different trade-offs for the various parameters of interest. One exception is the $\ell_{\infty}/\ell_2$ guarantee, for which the breakthrough work of Larsen, Nelson, Nguy{\^e}n and Thorup \cite{lnnt16} shows that this trade-off is unnecessary.

\subsection{Previous work}
 Since compressed sensing has been extensively studied in the literature for more than a decade, different guarantees of interest have been suggested ($x_{-k}$ is the vector that occurs after zeroing out every $i$ that does not belong among the largest $k$ coordinates). In what follows $x \in \mathbb{R}^n$ is the vector we want to sketch, $x'$ is the approximation to $x$, $k$ is the sparsity and $\epsilon$ is the fineness of the approximation. 

There are two different guarantees researchers consider in compressed sensing, one is the for-all guarantee and the other is the for-each guarantee. In the for-all guarantee, one wants to design a sketch that gives the desired result for all vectors $ x \in\mathbb{R}^n$. In the for-each guarantee, one wants to design a distribution over sketches that gives the desired result for a fixed vector $x \in \mathbb{R}^n$. We note that $\ell_{\infty}/\ell_2$, $\ell_2/\ell_2$ are impossible in the for-all model, unless $\Omega(n)$ measurements are used \cite{cdd09}. The standard approach for the for-all guarantee is via RIP matrices, satisfying the so-called Restricted Isometry Property. In what follows, we will refer to the for-each model, unless stated otherwise.

\begin{multicols}{2}
\begin{itemize}
  \item $\ell_{2}/\ell_2: \|x-x'\|_{2} \leq (1+\epsilon)\|x_{-k}\|_2$.
  \item $\ell_{\infty}/\ell_2: \|x-x'\|_{\infty} \leq (1+\epsilon) \frac{1}{\sqrt{k}}\|x_{-k}\|_2$.
  \item $\ell_{1}/\ell_1: \|x-x'\|_{1} \leq (1+\epsilon)\|x_{-k}\|_1$.
  \item $\ell_{2}/\ell_1: \|x-x'\|_{2} \leq (1+\epsilon)\frac{1}{\sqrt{k}}\|x_{-k}\|_1$.
\end{itemize}
\end{multicols}
The first set of schemes that initiated the research on compressed sensing are given in \cite{crt06,d06}. There the authors show, for any $x \in \mathbb{R}^n$, given $y = \Phi x$, it is possible to satisfy the $\ell_2/\ell_1$ guarantee for all vectors, if $\Phi$ is a Gaussian matrix with $O(k \log(n/k))$ rows. The schemes in \cite{cm06,ccf02} achieve the $\ell_{\infty}/\ell_2$ guarantee with $O(k \log n)$ measurementz, matching known lower bounds \cite{jst11}, $O(n \log n)$ decoding time and $O(\log n)$ update time.
The state of the art for $\ell_{\infty}/\ell_2$ is \cite{lnnt16}, which gives optimal number of measurements, sublinear decoding time, $O(\log n)$ update time and $1/\poly(n)$ failure probability. Price and Woodruff \cite{pw11} show that in order to get $\ell_2/\ell_2$ with constant failure probability$<1/2$ with the output being exactly $k$-sparse output requires $\Omega(\eps^{-2} k )$ measurements. They also showed non-$k$-sparse output requires $\Omega( \epsilon^{-1} k \log (n / k))$ measurements in the regime $\epsilon > \sqrt{k \log n / n}$, and gave an upper bound of $O(\epsilon^{-1} k \log n)$ measurements, showing thus a separation in the measurement complexity between $k$-sparse and $O(k)$-sparse output. Later, in the breakthrough work of Gilbert, Li, Porat and Strauss \cite{glps12} an algorithm that runs in sublinear time, and has $O(\log (n/k) \log^2k)$ column sparsity, was devised.
On generic norms, nearly optimal bounds have been given by Backurs, Indyk, Razenshteyn and Woodruff \cite{birw16}. We note, however, that their schemes are not computationally efficient: they have exponential running time, except in the case of Earth-Mover-Distance, which has time polynomial in $n$ and $\log^k n$.


\paragraph{Measurements.}

The number of measurements corresponds to physical resources: memory in monitoring devices of data streams, number of screens in biological applications, or number of filters in dynamic spectrum access (DSA) of radio signal \cite{huang2013applications}.

In applications such as medical imaging, it is crucial to reduce the number of measurements, since the radiation used in CT scans could potentially increase cancer risks for patients. 
For instance, \cite{psl+12} showed that a positive association between radiation exposure from CT scans in childhood and subsequent risk of leukemia and brain tumors.


For more applications, we refer the readers to \cite{qaisar2013compressive}. 

\paragraph{Encoding Time.} Designing algorithms with fast update/encoding time is a well-motivated task for streaming algorithms, since the packets arrive at an extremely fast rate \cite{tz12}; even logarithmic factors are crucial in these applications. Also in digital signal processing applications, in the design of cameras or satellites which demand rapid imaging, when we observe a sequence of images that are close to each other, we may not need to encode the new signal from the beginning, rather than encode only that part which differs from the current signal; the delay is then defined by the update time of our scheme. Moreover, in Magnetic Resonance Imaging (MRI) update time or encoding time defines the time the patient waits for the scan to happen. Improvement of the runtime has benefits both for patients and for healthcare economics \cite{ldsp08}.

A natural question is the following: what are the time limitations of our data structures, regarding update time? Regarding the streaming setting, the first lower bounds are given in \cite{lnn15} for non-adaptive algorithms. An algorithm is called non-adaptive if, during updates, the memory cells are written and read depend only on the index being updated and the random coins tossed before the stream is started to being processed. The lower bounds given concern both randomized and deterministic algorithms; the relevant bounds to sparse recovery are for $\ell_p/\ell_q$ estimation. However, for constant failure probability their results do not give anything useful, since their lower bounds start to kick in when the failure probability becomes very small, namely $o ( 2^{ -\sqrt{m \cdot \log n}} )$.

For the column sparsity (which could be smaller than update time, and hence the lower bounds in \cite{lnn15} might not apply\footnote{ the lower bounds in \cite{lnn15} also depend heavily on the streaming model, so they do not transfer necessarily to all scenarios where sparse recovery finds application.}), the only known lower bounds are known for RIP matrices, which are used in the for-all setting. To the best of our knowledge, the first non-trivial lower bounds were given by Nachin \cite{nachin2010lower}, and then extended by Indyk and Razenshteyn in \cite{ir13} for RIP-1  model-based compressed sensing matrices. Lower bounds for the column sparsity of RIP-$2$ matrices were given in Nelson and Nguy{\^e}n \cite{nelson2013sparsity}, and then to RIP-$p$ matrices in Allen-Zhu, Gelashvili and Razenshteyn \cite{agr16}. Roughly speaking, the lower bounds for $\ell_2$ indicate that if one aims for optimal measurements, $ m = k\log(n/k)$, in the regime $ k < n / \log^3n$, one cannot obtain column sparsity better than $\Omega( m )$. This indicates that the for-all case should be significantly worse, in terms of column sparsity, than the for-each case.

\paragraph{Decoding Time.} Another very important quantity we want to minimize is the time needed to reconstruct the approximation of $x$ from its compressed version. This quantity is of enormous significance in cases where the universe size is huge and we cannot afford to iterate over it. This is often the case in networking applications, where the universe size is the number of distinct IP addresses. In MRI applications the decoding time corresponds to the time needed to reconstruct the image after the scan has been performed. Decoding time is highly important also in satellite systems, modern radars and airspace surveillance, where compressed sensing have found extensive application \cite{radar}.


The sparse Fourier transform can be regarded as another variant of the sparse recovery problem, the results developed along that line being incomparable with traditional sparse recovery results, where we have freedom over the design of the sketching matrix. Sparse Fourier transforms can be divided into two categories, one being sparse discrete Fourier transform, and the other being sparse continuous Fourier transform \cite{bcgls14,m15,ps15,ckps16}. The sparse discrete Fourier transform also can be split into two lines, the first category of results \cite{ggims02,gms05,hikp12a,hikp12b,iw13,ikp14,ik14,k16,cksz17,k17,kvz19} carefully choose measurements that allow for sublinear recovery time, and the second category of results \cite{crt06a,rv08,b14,hr16} focus proving Restricted Isometry Property. The techniques used in sparse Fourier transforms are related to standard compressed sensing, since the main approaches try to implement arbitrary linear measurements via sampling Fourier coefficients.

%% file: tech.tex
\subsection{Our result}

Our main result is a novel scheme for $\ell_2/\ell_2$ sparse recovery. Our contribution lies in obtaining better decoding time, and $O(\log(n/k))$ column sparsity via new techniques. The problem of improving the column sparsity to $O(\log(n/k))$ was explicitly stated in \cite{glps12} as an open problem. We completely resolve the open problem. Moreover, as an important technical contribution, we introduce a different approach for sublinear-time optimal-measurement sparse recovery tasks. Since this iterative loop is a crucial component of almost all algorithms in sublinear-time compressed sensing \cite{ipw11,ps12,hikp12a, gnprs13,ikp14,k16,glps17,cksz17,k17,lnw17,nswz18}, we believe our new approach and ideas will appear useful in the relevant literature, as well as be a starting point for re-examining sparse recovery tasks under a different lens, and obtaining improved bounds.

\subsection{Notation}

For $ x \in \mathbb{R}^n$ we let $H(x,k)$ to be the set of the largest $k$ in magnitude coordinates of $x$. We also write $x_S$ for the vector obtained after zeroing out every $x_i, \notin S$, and $x_{-k} = x_{[n] \setminus H(x,k)}$. We use $\|\cdot\|_p$ to denote the $\ell_p$ norm of a vector, i.e. $\|x\|_p = \left( \sum_{i=1}^n |x_i|^p \right )^{1/p}$.  

\subsection{Technical statements}

\begin{table}[!t]
\begin{center}
    \begin{tabular}{| l | l | l | l| }
    \hline
    {\bf Reference} & {\bf Measurements} & {\bf Decoding Time} & {\bf Encoding Time}   \\ \hline 
    \cite{d06,crt06} & $k \log (n/k)$ & LP & $k \log (n/k)$ \\ \hline
    \cite{ccf02, cm06} & $\epsilon^{-2} k \log n $ & $\epsilon^{-1} n \log n$ & $\log n$ \\ \hline
    \cite{nt08} & $ k\log (n/k)$ & $nk\log (n/k)$ & $\log(n/k)$ \\ \hline
    \cite{cm04} & $\epsilon^{-2}k \log^2 n$ & $\epsilon^{-1}k \log^c n$ & $\log^2 n$ \\ \hline
    \cite{ccf02, cm06} & $\epsilon^{-2} k \log^c n $ & $\epsilon^{-1} k \log^2 n$ & $\log^c n$ \\ \hline
    \cite{glps12} & $\epsilon^{-1} k \log(n/k) $  &  $\epsilon^{-1} k \log^c n$ & $ \log (n/k) \cdot \log^2 k   $   \\ \hline 
    Our result & $\epsilon^{-1} k  \log(n/k)$ &  $\epsilon^{-1} k  \log^2(n/k)$ & $\log(n/k)$ \\\hline
    \end{tabular}
\end{center}
\caption{(A list of $\ell_2/\ell_2$-sparse recovery results). We ignore the ``$O$'' for simplicity. LP denotes the time of solving Linear Programs \cite{cls19}, and the state-of-the-art algorithm takes $n^{\omega}$ time where $\omega$ is the exponent of matrix multiplication. The results in \cite{d06,crt06,nt08} do not explicitly state the $\ell_2/\ell_2$ guarantee, but their approach obtains it by an application of the Johnson-Lindenstrauss Lemma; they also cannot facilitate $\epsilon < 1$, obtaining thus only a $2$-approximation. The $c$ in previous work is a sufficiently large constant, not explicitly stated, which is defined by probabilistically picking an error-correcting code of short length and iterating over all codewords. We estimate $ c \geq 4$. We note that our runtime is (almost) achieved } 
\label{tab:ell2_sparse_recovery} 
\end{table}


We proceed with the definition of the $\ell_2/\ell_2$ sparse recovery problem.
\begin{problem}[$\ell_2/\ell_2$ sparse recovery]
Given parameters $\epsilon, k, n$, and a vector $x \in \R^n$.  The goal is to design some matrix $\Phi \in \R^{m \times n}$ and a recovery algorithm ${\cal A}$ such that we can output a vector $x'$ based on measurements $y = \Phi x$,
\begin{align*}
\| x' - x \|_2 \leq (1+\epsilon) \min_{k\textrm{-}\sparse~z \in \mathbb{R}^n} \| z - x \|_2.
\end{align*} 
We primarily want to minimize $m$ (which is the number of measurements), the running time of ${\cal A}$ (which is the decoding time) and column sparsity of $\Phi$. \end{problem}

In table~\ref{tab:ell2_sparse_recovery}, we provide a list of the previous results and compare with ours. Here, we formally present our main result.
\begin{theorem}[stronger $\ell_2/\ell_2$ sparse recovery] \label{thm:main_sparse_recovery}
There exists a randomized construction of a linear sketch $\Phi \in \mathbb{R}^{m \times n}$ with $m = O(\epsilon^{-1} k \log (n/k)  )$ and column sparsity $O( \log (n/k) )$, such that given $y = \Phi x$, we can find an $O(k)$-sparse vector $x' \in \mathbb{R}^n$ in $O(m \cdot \log(n/k))$ time such that 
\begin{align*}
\| x' - x \|_2 \leq (1+\epsilon) \min_{k\textrm{-}\sparse~z \in \mathbb{R}^n} \| z - x \|_2.
\end{align*}
holds with $9/10$ probability. 
\end{theorem}

\begin{remark}
In the regime where $k$ is very close to $n$, for example $ k = n / \mathrm{poly}(\log n)$, we get an exponential improvement on the column sparsity over \cite{glps12}. In many applications
of compressed sensing, this is the desired regime of interest, check for example Figure 8 from \cite{bi08}: $n = 71,542$ while $m \geq 10,000$, which corresponds to $k$ being very close to $n$.
\end{remark}

\begin{remark}
As can be inferred from the proof, our algorithms runs in time $O( (k/\epsilon) \log^2(\epsilon n/k) + (k/\epsilon) \log(1/\epsilon))$, which is slightly better than the one stated in Theorem \ref{tab:ell2_sparse_recovery}. The algorithm in \cite{hikp12a} achieves also the slightly worse running time of $O( (k/\epsilon) \log n \log (n/k) )$. That algorithm was the first algorithm that achieved running time $O(n \log n)$ for all values of $k,\epsilon$ for which the measurement complexity remained sublinear, smaller than $\gamma n$, for some absolute constant $\gamma$. A careful inspection shows that our algorithm achieves running time that is \textbf{always} sublinear, as long as the measurement complexity remains smaller than $\gamma n$.
\end{remark}

\subsection{Overview of techniques and difference with previous work}
This subsection is devoted to highlighting the difference between our approach and the approach of \cite{glps12}. We first give a brief high-level description of the state of the art algorithm before our work, then discuss our techniques, and try to highlight why the previous approach could not obtain the stronger result we present in this paper. Lastly, we show how our ideas can be possibly applied to other contexts.

\subsubsection{Summary of \cite{glps12}.} The algorithm of \cite{glps12} consists of $O(\log k)$ rounds: in the $r$-th round the algorithm finds a constant fraction of the remaining heavy hitters. Beyond this iterative loop lies the following idea about achieving the $\ell_2/\ell_2$ guarantee: in order to achieve it, you can find all but $\frac{k}{3^r}$ heavy hitters $i$ such that $ |x_i|^2 = \Omega(\frac{2^r \epsilon}{k} \|x_{-k}\|_2^2)$.  This means that the algorithm is allowed to ``miss'' a small fraction of the heavy hitters, depending on their magnitude. For example, if all heavy hitters are as small as $\Theta(\sqrt{\epsilon/k} \|x_{-k}\|_2)$, a correct algorithm may even not find any of them. This crucial observation leads naturally to the main iterative loop of combinatorial compressed sensing, which, as said before, loop proceeds in $O(\log k)$ rounds. Every round consists of an identification and an estimation step: in the identification step most heavy hitters are recognized, while in the estimation step most of them are estimated correctly. Although in the estimation step some coordinates might have completely incorrect estimates, this is guaranteed (with some probability) be fixed in a later round. The reason why this will be fixed is the following. If a coordinate $i$ is badly estimated, then it will appear very large in the residual signal and hence will be identified in later rounds, till it is estimated correctly. One can observe that the correct estimation of that round for coordinate $i$ cancels out (remedies) the mistakes of previous rounds on coordinate $i$. Thus, the identification and estimation procedure, which are interleaved, work complementary to each other. The authors of \cite{glps12} were the first that carefully managed to argue that identifying and estimating a constant fraction of heavy hitters per iteration, gives the optimal number of measurements.

More formally, the authors prove the following iterative loop invariant, where $k_r = k3^{-r}, \epsilon_r = \epsilon_r2^{-r}$ for $ r \in [R]$ with $R = \log k$:
Given $ x \in \mathbb{R}^n$ there exists a sequence of vectors $\{ x^{(r)} \}_{r \in [R]} $, such that $x^{(r+1)} = x^{(r)} - \hat{x}^{(r)}$ and 
\begin{align}\label{eq:iterative_procedure}
\|(x - \hat{x})_{-k_r}\|_2^2 \leq ( 1+\epsilon_r) \|x_{-k_r}\|_2^2.
\end{align}

In the end, one can apply the above inequality inductively to show that 
\begin{align*}
\|x - \sum_{r =1}^R \hat{x}^{(r)} \|_2^2 \leq ( 1+\epsilon)\|x_{-k}\|_2^2.
\end{align*}

We now proceed by briefly describing the implementations of the identification and the estimation part of \cite{glps12}. In the identification part, in which lies the main technical contribution of that work, every coordinate $i \in [n]$ is hashed to $O(k/\epsilon)$ buckets and in each bucket $O(\log(\epsilon n/k))$-measurement scheme based on error-correcting codes is used to identify a heavy hitter; the authors carefully use a random error-correcting code of length $O(\log \log(\epsilon n/k))$, so they afford to iterate over all codewords and employ a more sophisticated approach and use nearest-neighbor decoding. This difference is one of the main technical ideas that allow them to obtain $O(\epsilon^{-1}k \log(n/k))$ measurements, beating previous work, but it is also the main reason why they obtain $k\cdot \poly(\log n)$ decoding time\footnote{The authors do not specifically address the exponent in the $\poly(\log n)$, but we estimate it to be $\geq 4$.}: performing nearest neighbor decoding and storing the code per bucket incurs additional $\poly(\log n)$ factors. Moreover, for every iteration $r$ they need to repeat the identification scheme $r$ times in order to bring down the failure probability to $2^{-r}$, so that they afford a union-bound over all iterations. This leads to an additional $O(\log^2 k)$ factor in the update time. The estimation step consists of hashing to $O(k/\epsilon)$ buckets and repeating $O(\log(1/\epsilon))$ times. Since the identification step returns $O(k/\epsilon)$ coordinates, the $O(\log(1/\epsilon))$ repetitions of the estimation step ensure that at most $k / 3$ coordinates out of the $O(k/\epsilon)$ will not be estimated correctly. This is a desired property, since it allows the algorithm to keep the $2k$ coordinates with the largest estimates, subtract them from $x$ and iterate.

In the next section, we will lay out our approach which improves the decoding time and the column sparsity of \cite{glps12}. The iterative procedure of \cite{glps12} lies in the heart of most compressed sensing schemes, so we believe that this new approach could be applied elsewhere in the sparse recovery literature.

\subsubsection{Our approach}\label{sec:intro_btree}

As we mentioned before, our approach is totally different from previous work, avoiding the iterative loop that all algorithms before applied.
Our algorithm consists of four steps, each one being a different matrix responsible for a different task. The first matrix, with a constant number of rows allows us to approximate the tail of the vector $x$, an approximation that will appear useful in the next step. The second matrix along with its decoding procedure, which should be regarded as the identification step, enables us to find a list $L$ of size $O(k/\epsilon)$ that contains $k$ coordinates\footnote{We note that this term is exactly $k$, not $O(k)$. Although not important for our main result, it will be crucial for some of our applications of our techniques.}, which are sufficient for the $\ell_2/\ell_2$ guarantee. This matrix has $O(\epsilon^{-1} k \cdot \log( \epsilon n / k)) )$ rows. The third matrix with its decoding procedure, which should be regarded as the pruning step, takes the aforementioned list $L$, and prunes it down to $O(k)$ coordinates, which are sufficient for the $\ell_2/\ell_2$ guarantee. This matrix again has $O(\epsilon^{-1} k \cdot \log(1/\epsilon))$ rows, for a total of $O(\epsilon^{-1} k \log(n/k))$ rows. The last matrix is a standard set-query sketch.


\paragraph{Step 1: Tail Estimation}

\begin{lemma}[tail estimation]\label{lem:tail_estimation}
Let $c_1 \geq 1$ denote some fixed constant. There is an oblivious construction of matrix $\Phi \in \R^{m \times n}$ with $m = O(\log (1/\delta))$ and column sparsity $O(\log(1/\delta))$ such that, given $\Phi x$, there is an algorithm that outputs a value $V \in \R$ in time $O(m)$ such that
\begin{align*}
\frac{1}{10k} \| x_{- c_1 \cdot k } \|_2^2 \leq V \leq \frac{1}{k} \| x_{-k} \|_2^2
\end{align*}
holds with probability $1-\delta$.
\end{lemma}

 Our first step towards the way for stronger sparse recovery is the design a routine that estimates the $\ell_2$ norm of the tail of a vector $ x \in \mathbb{R}^n$, which we believe might be interesting in its own right. More generally, our algorithm obtains a value $V$ such that
$
\frac{1}{10k} \|x_{-c_1 \cdot k}\|_p^p \leq V \leq \frac{1}{k}\|x_{-k}\|_p^p,
$
using $O(1)$ measurements. Here $c_1$ is some absolute constant. 
To obtain this result we subsample the vector at rate $\Theta(1/k)$ and then use a $p$-stable distribution to approximate the subsampled vector.  While the upper bound is immediate, the Paley-Zygmund inequality does not give a sufficient result for the lower bound, so more careful arguments are needed to prove the desired result. We obtain our result by employing a random walk argument. 

One additional possible application in sparse recovery  applications where a two-stage scheme is allowed, e.g. \cite{dlty06,ddtlskb08}, would be to first use the above routine to roughly estimate how many heavy coordinates exist, before setting up the measurements. For example, we could first run the above routine for $k =1 , 2 ,2^2,\ldots, 2^{\log n}$, obtaining values $V_1,V_2,V_4,\ldots, V_{\log n}$, and then use these values to estimate the size of the tail of the vector is, or equivalently approximate the size of the set of the heavy coordinates by a number $k'$. We can then run a sparse recovery algorithm with sparsity $k'$. Details can be found in Section~\ref{sec:tail_estimation}.

\paragraph{Step 2: The Identification Step}

The goal of this step is to output a list $L$ of size $O(k/\epsilon)$ that contains a set of $k$ coordinates that are sufficient to satisfy the $\ell_2/\ell_2$ guarantee. The column sparsity we are shooting for at this point is $O(\log (\epsilon n / k))$ 
, and the decoding time should be $O(m \log(\epsilon n/k))$.

\begin{lemma}[identification sketch]\label{lem:identification_sketch}
There exists a randomized construction of a matrix $\Phi \in \mathbb{R}^{m \times n}$, with $m = O(\epsilon^{-1} k \cdot \log (\epsilon n / k) )$ and column sparsity $O(\log (\epsilon n / k))$, such that given $ y= \Phi x$, one can find a set $L$ of size $O(k/\epsilon)$ in $O(m \log(\epsilon n/k))$ time, such that
\begin{align*} 
\exists T \subset S, |T| \leq k: \|x - x_T \|_2 \leq ( 1+\epsilon) \|x_{-k}\|_2,
\end{align*}
holds with probability at least $9/10$.
\end{lemma}

For this routine, we will set up a hierarchical separation of $[n]$ to trees. We will call this separation interval forest. we set $ \tau = k/\epsilon$. Then we partition $[n]$ into $\tau$ intervals of length $q = n / \tau$, and set up an interval tree for each interval in the following way: every interval tree has branching factor 
\begin{align*}
	\frac{\log q}{\log \log q}
\end{align*}
with the same height (this is consistent with the fact that every tree contains $q$ nodes). At the leaves of every interval tree there are the nodes of the corresponding interval.

Our approach consists is now the following. For each level of the interval forest we hash everyone to $k/\epsilon$ buckets in the following way: if two coordinates $i,i'$ are in the same interval (node of the interval forest) they are hashed to the same bucket. The above property can be regarded as ``hashing interval to buckets''. Moreover, every $x_i$ is multiplied by a \textbf{Gaussian} random variable. We repeat this process for $\log \log (\epsilon n/k)$ per level.

The decoding algorithm is the following. First, we obtain a value $V$ using the routine in Step $1$, Lemma \ref{lem:tail_estimation}. Then we proceed in a breadth-first (or depth-first, it will make no difference) search manner and find an estimate for every interval, similarly to \cite{lnnt16,ch09}, by taking the median of the $\log\log (\epsilon n /k)$ buckets it participates to. 
 There are two technical things one needs to show: First, we should bound the decoding time, as well as the size of the output list $L$. Second, we need to show that there exists a set $T'$ of size at most $k$ that satisfies the guarantee of the Lemma \ref{lem:identification_sketch}. For the first part, we show that the branching process defined by the execution of the algorithm is bounded due to the $\log \log(\epsilon n/k)$ repetitions per level. For the second part, we show that for every coordinate $i \in H(x,k)$ the probability that $i \in L$ is proportional to 
$ k|x_i|^2 / ( \epsilon \|x_{-k}\|_2^2 )$. Then we show that the expected $\ell_2^2$ mass of coordinates $ i \in L \setminus H(x,k)$ is $\epsilon \|x_{-k}\|_2^2$. This suffices to give the desired guarantee for Lemma \ref{lem:identification_sketch}.

We provide details in Section~\ref{sec:btree}.

\paragraph{Step 3: The Pruning Step}

\begin{lemma}[pruning sketch]\label{lem:pruning_sketch}
Let $c_2, c_3> 1$ denote two fixed constants. There exists a randomized construction of a matrix $\Phi \in \mathbb{R}^{m\times n}$, with $m = O( \epsilon^{-1} k \cdot \log (1/\epsilon)),$ with column sparsity $O(\log (1/\epsilon))$ such that the following holds :\\
Suppose that one is given a (fixed) set $L \subseteq [n] $ such that 
\begin{align*}
 |L| = O(k/\epsilon), && \exists T \subset L, |T| \leq k: \|x - x_T \|_2 \leq ( 1+\epsilon) \|x_{-k}\|_2.
\end{align*}
Then one can find a set $S$ of size $c_2 \cdot k$  in time $O(m)$, such that
\begin{align*}
\|x - x_{S} \|_2 \leq (1 + c_3 \cdot \epsilon) \|x_{-k}\|_2
\end{align*}
holds with probability $9/10$.
\end{lemma}

 We will now prune the list $L$ obtained from the previous step, to $O(k)$ coordinates. We are going to use $O(\epsilon^{-1} k \cdot \log(1/\epsilon))$ measurements. We hash every coordinate to $k/\epsilon$ buckets, combining with \textbf{Gaussians}, and repeating $O(\log(1/\epsilon))$ times. A similar matrix was also used in \cite{glps12}, but there the functionality and the analysis were very different; moreover, the authors used random signs instead of Gaussians. We will heavily exploit the fact that a standard normal $g$ satisfies $\Pr[|g| < x] = O(x)$. The reasons why we need Gaussians is the following: if we have a lot of large coordinates which are equal and much larger than $\sqrt{\epsilon} \|x_{-k}\|_2$, we want to find all of them, but due to the use of random signs, they might cancel each other in a measurement. Switching to Gaussians is the easiest way of avoiding this undesirable case.

Our algorithm computes, for every $i \in L$, an estimate $\hat{x}_i$ by taking the median of the $O(\log(1/\epsilon))$ buckets it participates to, and then keeps the largest $O(k)$ coordinates in magnitude to form a set $S$. We then show that these coordinates satisfy the $\ell_2/\ell_2$ guarantee. We will say that a coordinate is well-estimated if $ |\hat{x}_i|  = \Theta( |x_i| ) \pm \sqrt{\epsilon k^{-1} } \|x_{-k}\|_2$. 
 For the analysis, we define a threshold $\tau = \|x_{-k}\|_2/ \sqrt{k}$, and classify coordinates based on whether $|x_i| \geq \tau$ or not.
\begin{itemize}
\item In the case $|x_i| \geq \tau$ the expected mass of these coordinates $i \notin S$ is small;
\item In the other case the number of coordinates $i$ with $|x_i| < \tau$ are $O(k)$. This allows us to employ an exchange argument, similar to previous work, e.g. \cite{pw11}, but more tricky due to the use of Gaussians instead of random signs. 
\end{itemize}
We note that the way we compute the estimates $\hat{x}_i$ is different from previous work: one would expect to divide the content of a bucket that $i$ hashed to by the coefficient assigned to $x_i$, in order to get an unbiased estimator, but this will not work. The details can be found in Section~\ref{sec:pruning_sketch}.

In the end, this step gives us a set $S$ suitable for our goal, but \textbf{does not} give good estimations of the coordinates inside that set. For that we need another, standard step.

\paragraph{Step 4: Set Query} 

We estimate every coordinate in $S$ using a set query algorithm of Price \cite{p11}, obtaining the desired guarantee. This matrix needs only $O(k/\epsilon)$ measurements, and runs in $O(k)$ time, while having constant column sparsity.

\begin{lemma}[set query, \cite{p11}]\label{lem:set_query_price}
For any $\epsilon \in (0,1/10]$. There exists a randomized construction of a matrix $\Phi \in \mathbb{R}^{m \times n}$, with $m = O(k/\epsilon)$ and column sparsity $O(1)$, such that given $y = \Phi x$ and a set $S \subseteq [n]$ of size at most $k$, one can find a $k$-sparse vector $\hat{x}_S$, supported on $S$ in $O(m)$ time, such that 
\begin{align*}
 \|\hat{x}_S - x_S \|_2^2 \leq \epsilon \|x_{[n] \backslash S}\|_2^2.
\end{align*}
holds with probability at least $1 - 1/\poly(k)$.
\end{lemma}

Our Theorem \ref{thm:main_sparse_recovery} follows from the above four steps by feeding the output of each step to the next one. In the end, we rescale $\epsilon$. 
More specifically, by the identification step we obtain a set $L$ of size $O(k/\epsilon)$ which contains a subset of size $k$ that satisfies the $\ell_2/\ell_2$ guarantee. Then, the conditions for applying the pruning step are satisfied, and hence we can prune the set $L$ down to $O(k)$ coordinates, which satisfy the $\ell_2/\ell_2$ guarantee. Then we apply the set-query sketch to obtain estimates of these coordinates.

In what follows we ignore constant terms.
The number of measurements in total is
 \begin{align*}
\underbrace{ 1 }_{ \text{tail~estimation} } + \underbrace{ (k/\epsilon) \log (\epsilon n/k) }_{ \text{identification~step} } + \underbrace{ (k/\epsilon) \log (1/\epsilon) }_{ \text{pruning~step} } + \underbrace{ (k/\epsilon) }_{ \text{set~query} }.
\end{align*}

The decoding time equals
\begin{align*}
\underbrace{ 1 }_{ \text{tail~estimation} } + \underbrace{ (k/\epsilon) \log ( \epsilon n / k ) \log( \epsilon n/k) }_{ \text{identification~step} } + \underbrace{ (k/\epsilon) \log (1/\epsilon) }_{ \text{pruning~step} } + \underbrace{ k }_{ \text{set~query} }.
\end{align*}

The column sparsity equals 
\begin{align*}
\underbrace{ 1 }_{ \text{tail~estimation} } + \underbrace{ \log (\epsilon n / k)  }_{ \text{identification~step} } + \underbrace{ \log (1/\epsilon) }_{ \text{pruning~step} } + \underbrace{ 1 }_{ \text{set~query} }.
\end{align*}

\subsection{Possible applications of our approach to other problems}

\paragraph{Exactly $k$-sparse signals.}

When the vector we have to output has to be $k$-sparse, and not $O(k)$-sparse, the dependence on $\epsilon$ has to be quadratic \cite{pw11}. Our algorithm yields a state of the art result for this case, too. One can observe that the analysis of the algorithm in Section \ref{sec:btree} outputs a $O(k/\epsilon)$-sized set which contains a set $T$ of size $k$ that allows $\ell_2/\ell_2$ sparse recovery. Performing a \textsc{CountSketch} with $O(k/\epsilon^2)$ columns and $O( \log (k/\epsilon)$ rows, and following a standard analysis, one can obtain a sublinear algorithm with measurement complexity $O( \epsilon^{-1} k \log( \epsilon n /k) + \epsilon^{-2} k \log (k/\epsilon))$. This is an improvemnet on both the runtime and measurement complexity over previous work \cite{ccf02}.

\paragraph{Block-Sparse Signals.}

Our algorithm easily extends to block-sparse signals. For a signal of block size $b$ we obtain sample complexity $O( k /(\epsilon b ) \log ( b n / k) + (k/\epsilon))$, with a running time nearly linear in $k$. This matches the sample complexity of previous super-linear algorithms \cite{bcdh10,cevher2009recovery}, which also could not facilitate $\epsilon < 1$.


\paragraph {Phaseless Compressed Sensing.}
In Phaseless Compressed Sensing, one wants to design a matrix $ \Phi$ with $m$ rows, such that given $y = \Phi x$, one can find a vector $\hat{x}$ such that $ \mathrm{min}_{ \theta \in [0,2\pi]} \|x - e^{i \theta } \hat{x}\|_2 \leq (1+\epsilon)\|x_{-k}\|_2$. This problem has received a fair amount of attention \cite{ohlsson2011compressive,li2013sparse,cbjc14,ylpr15,pylr17,n17,ln17}, and the state of the art algorithm has $O(k \log n)$ measurement complexity \cite{n17,ln17}. One of the problems is that the iterative loop approach cannot be used here, since it is heavily based on the linearity of the sketch. However, our identification and pruning step do not use the linearity of the sketch, and work also with phaseless measurements. Previous algorithms such as \cite{n17,ln17} suffered a $k \log n$ factor in the number of measurements already from the first step, but this is avoidablue using our new approach. We hope to see an extension down this avenue that gives $O(k \log (n/k))$ measurements.

\paragraph{One-Bit Compressed Sensing.}

Another important subfield of Sparse Recovery is one-bit compressed sensing, where one has access only to one-bit measurements, i.e. $y = \mathrm{sign}(Ax)$, where the $\mathrm{sign}$ function on vectors should be understood as pointwise application of the sign function on each entry. Sublinear algorithms appear in \cite{nakos2017fast,nakos2017one}, but they both do not obtain the optimal number of measurements in terms of $k$ and $n$, which is $k \log (n/k)$, but rather the slightly suboptimal $k \log n$. One of the most important reasons is that the iterative loop cannot be implemented in such a scenario. It is a natural question whether our new approach can give the optimal number of measurements. The thresholding step, namely the part where we take use $V$ to filter out non-heavy intervals cannot be implemented here, but perhaps there is still a way to make a similar argument. One first approach should be to show that sublinear decoding with optimal measurements is achieved using non-adaptive threshold measurements, such as in \cite{knudson2016one} and \cite{baraniuk2017exponential} (note that the latter one uses adaptive measurements though).

\paragraph{Sparse Fourier Transform.}
The standard approach to discrete sparse Fourier transform, is to implement linear measurements by using Fourier measurements \cite{ggims02,gms05,hikp12a,hikp12b,iw13,ikp14,ik14,k16,cksz17,k17}. The idea is to hash the spectrum to $B$ buckets by carefully multiplying in the time-domain the vector $x$ with a sparse vector $z$. In the frequency domain this corresponds to convolving the spectrum of the $x$ with an approximation of the filter of an indicator function of an interval of length roughly $B$. Due to the Uncertainty Principle, however, one has to exchange measurement complexity and decoding time with the quality of the filter. For example, implementing hashing to buckets using ``crude'' filters leads to leakage in subsequent buckets, giving additional error terms. When iterating as usual, these errors accumulate and make identification much harder. The sophisticated approach of \cite{k17} manages to design an iterative algorithm, in the same vein with previous algorithms, which takes $O(\epsilon^{-1} k \log n)$ measurements. It would be interesting to see if the approach we suggest avoids some of the problems created by this iterative loop, and can give simpler and faster sparse Fourier transform schemes. It would be interesting to obtain such a result even using adaptive measurements. The work \cite{cksz17} has some interesting ideas in the context of block-sparse vectors that could be relevant.



%% file: btree.tex
\section{The Identification Linear Sketch}\label{sec:btree}
The goal of this section is to prove the following result,
\begin{theorem}[Restatement of Lemma~\ref{lem:identification_sketch}]\label{thm:identification}
Let $C_L > 1$ be a fixed constant. There exists a randomized construction of a matrix $\Phi \in \mathbb{R}^{m \times n}$, with 
 \[  m = O(\epsilon^{-1}k \log(\epsilon n/k)),\] 
with column sparsity $O(\log(\epsilon n/k))$ such that given $ y= \Phi x$, one can find in time $O(m \log( n/k) )$ a set $L$ of size $C_L \cdot k/\epsilon$, such that
\begin{align*} 
\exists T \subset L, |T| \leq k: \|x - x_T \|_2 \leq ( 1+\epsilon) \|x_{-k}\|_2,
\end{align*}
with probability $9/10$.
\end{theorem}


In Section~\ref{sec:btree_algorithm}, we provide the definition of sketching matrix $\Phi$ and present the decoding algorithm. We proved some concentration result in Section~\ref{sec:btree_concentration_of_estimation}. We analyzed the running time of algorithm in Section~\ref{sec:btree_running_time}. We proved the guarantees of the algorithm in Section~\ref{sec:btree_guarantee}. Finally, we bound the number of measurements in Section~\ref{sec:btree_measurements}. 
\subsection{Design of the sketch and decoding algorithm}\label{sec:btree_algorithm}





\begin{table}[!h]
\centering
\begin{tabular}{|l|l|l|l|}
  \hline
  Notation & Choice & Statement & Parameter\\ \hline
  $C_H$ & $4$  & Definition~\ref{def:degree_height} & $H$ \\ \hline
  $C_R$ & $100$ & Definition~\ref{def:R} & $R$ \\ \hline
  $C_B$ & $10^5$ & Definition~\ref{def:sketching_matrix} & $B$ \\ \hline
  $C_0$ & $10^3$ & Lemma~\ref{lem:lplp_tail_estimation} & Blow up on tail size \\ \hline
  $C_L$ & $10^4$ & Lemma~\ref{lem:running_time} & $L$ \\ \hline
  $\eta$ & $1/9$ & Lemma~\ref{lem:large},\ref{lem:small} & Shrinking factor on $V$ \\ \hline
  $\zeta$ & $1/4000$ & Lemma~\ref{lem:large},\ref{lem:small} & $\zeta \leq \eta / 400$ \\ \hline
\end{tabular}\caption{Summary of constants in Section~\ref{sec:btree}, the column ``Parameter'' indicates which parameter is depending on that constant. Note that constants $C_H, C_R, C_B, C_0, \eta$ are used in both algorithm and analysis, but constants $C_L$ and $\zeta$ are only being used in analysis. $C_L$ is the related to the guarantee of the output of the algorithm.}\label{tab:btree} 
\end{table}

We are going to use a hierarchical separation of $[n]$ into intervals. We will call this separation an interval forest.

Before discussing the matrix, we need the following definitions. We define  
\begin{definition}[size of the each tree in the forest]\label{def:forest}
Let 
\begin{align*}
\tau = k/\epsilon,
\end{align*}
assuming that $k/\epsilon \leq n/16$.
The size of each tree in the forest is
\begin{align*}
q = n /\tau = n\epsilon / k
\end{align*}
\end{definition}

\begin{definition}[degree and height of the interval forest]\label{def:degree_height}
Let $C_H > 1$ be a sufficiently large constant such that
\begin{align*}
 \left( \log q / \log \log q \right)^{C_H \log q /\log \log q} \geq q.
\end{align*}
Let $D$ denote the degree of the tree, and let $H$ denote the height of the tree. We set $D$ and $H$, $D = \lceil \log q / \log \log q \rceil$, and $H =  \lceil C_H \log q / \log \log q \rceil$.
\end{definition}

\begin{definition}[number of repetitions per level] \label{def:R}
Let $R$ denote the number of repetitions in each level. Let $C_R > 1$ denote some sufficiently large constant. We set $R = C_R \log \log q$.
\end{definition}

For $\ell \in \{0,1,\ldots,H\}$ we define $\mathcal{I}_{\ell}$, which is a family of sets. Every set $\mathcal{I}_\ell$ is a decomposition of $[n]$ to $\tau D^\ell$ intervals of (roughly) the same length. The set $ \mathcal{I}_0$ is a decomposition of $[n]$ to $\tau$ intervals of (roughly) the same length length $q$. If needed, we can round $q$ to a power of $2$. This means that
\begin{align*}
\mathcal{I}_0 = \{ I_{0,1} , I_{0,2}, \ldots\}
\end{align*}
where
\begin{align*}
I_{0,1} = [1,  \tau], I_{0,2} = [\tau+1,  2\tau], \ldots
\end{align*} 


We can conceptualize these sets as a forest consisting of $\tau$ trees, each of branching factor $D$ and height $H$, where the $\ell$-th level partitions $[n]$ into disjoint $\tau D^\ell$ intervals of length $n/(\tau \cdot D^{\ell}) = q \cdot D^{-\ell}$. For $\ell \in \{0, \ldots, H\}$, interval $I_{\ell,j}$ is decomposed to $D$ disjoint and continuous intervals 
\begin{align*}
I_{\ell+1,j\cdot D +1}, \ldots, I_{\ell+1,(j+1) \cdot D}
\end{align*}
 of the same length, except possibly the last interval.

 We say that an interval $I_{\ell+1,j}$ is a child of interval $I_{\ell,j'}$ if $j' = \lceil j/D \rceil$. 

\begin{definition}[sketching matrix $\Phi$]\label{def:sketching_matrix}
Let $C_B > 1$ be a sufficiently large constant. Let $B = C_B k / \epsilon$.
Let matrices $\Phi^{(1)},\ldots, \Phi^{(H)}$, where every matrix $\Phi^{(\ell)}$ consists of $R $ submatrices $\{\Phi^{(\ell)}_r\}_{r \in [R]}$. For every $\ell \in [H]$ and $r\in[ R ]$, we pick $2$-wise independent hash functions $h_{\ell,r}: [\tau D^{\ell}] \rightarrow [ B ]$.

We define measurement $y_{\ell,r,b} = (\Phi^{(\ell)}_r x)_b $ as:
\begin{align*}
y_{\ell,r,b} = \sum_{j \in h_{\ell,r}^{-1}(b) } \sum_{i \in I_{\ell,j}} g_{i,\ell,r}x_j,
\end{align*} 
where $g_{i,\ell,r} \sim \mathcal{N}(0,1)$, i.e. independent standard Gaussians.

We slightly abuse notation and treat $y$ as matrix the mapping to vector should be clear.
\end{definition}

Note that $C_B$ should be chosen such that $C_B \gg C_0$, where $C_0$ appears in tail estimation in Lemma~\ref{lem:lplp_tail_estimation}.

\begin{algorithm}[!t]
\begin{algorithmic}[1]
\Procedure{\textsc{IntervalForestSparseRecovery}}{$x,n,k,\epsilon $} \Comment{Theorem~\ref{thm:identification}} 
	\State Choose constants $C_H$, $C_R$, $\eta$, $C_0$ \Comment{According to Table~\ref{tab:btree}}
	\State $\tau \leftarrow (k/\epsilon)$ \Comment{Definition~\ref{def:forest}}
	\State $q \leftarrow n / \tau$ \Comment{Definition~\ref{def:forest}}
	\State $H \leftarrow \lceil C_H \log q /  \log \log q \rceil $ \Comment{Definition~\ref{def:degree_height}}
	\State $D \leftarrow \lceil \log q/ \log \log q \rceil$ \Comment{Definition~\ref{def:degree_height}}
	\State $R \leftarrow \lceil C_R \log \log q \rceil$ \Comment{Definition~\ref{def:R}} 
	\State $\ov{V} \leftarrow \textsc{LpLpTailEstimation}( x , k, 2, C_0, 1/100 )$ \label{lin:ov_v_value}\Comment{Algorithm~\ref{alg:lplp_tail_estimation}}
	\State $V \leftarrow \epsilon \ov{V}$ \label{lin:v_value} \Comment{Lemma~\ref{lem:large}}
	\State $T_0 \leftarrow \{ I_{0,1}, \ldots, I_{0,\tau} \}$
	\For{$\ell=1 \to H$}\label{lin:beg_of_btree}
		\State $T_{\ell} \leftarrow\textsc{RecursiveBTree}(\ell,R,D,\eta,T_{\ell-1},V)$  \Comment{Lemma~\ref{lem:running_time}}
	\EndFor\label{lin:end_of_btree} 
	\State $L \leftarrow T_H$
	\State \Return $L$ \Comment{Lemma~\ref{lem:final}}
\EndProcedure
\Procedure{\textsc{RecursiveBTree}}{$\ell,R,D,\eta,T,V$} 
	\State $T' \leftarrow \emptyset$
	\For{$t \in T$}
		\State Let $I_{\ell,j_1}, I_{\ell,j_2}, \ldots, I_{\ell,j_{D}}$ denote the child intervals of $I_{\ell-1,t}$
		\For {$p \in [D]$}
			\State $z_{j_p} \leftarrow \mathrm{median}_{r\in [R]} |y_{\ell,r,h_{\ell,r}(j_p) }|^2$ \Comment{Definition~\ref{def:sketching_matrix}}
		\EndFor
		\If { $z_{j_p} \geq \eta V$} \label{lin:btree_z_j_p_geq_eta_V}
			\State $T' \leftarrow T' \cup \{  j_p \}$
		\EndIf
	\EndFor
	\State \Return $T'$
\EndProcedure
\end{algorithmic}\caption{interval-forest sparse recovery}\label{alg:btree}
\end{algorithm}

\subsection{Concentration of estimation}\label{sec:btree_concentration_of_estimation}
In the following lemmata, $\zeta, \eta \in (0,1)$ are absolute constants with $1 > \eta >  1/10 >  \zeta$ (See a choice for our application in Table~\ref{tab:btree}). The exact values of the constants will be chosen below.

The following lemma handles the probability of detecting a heavy interval at level $\ell$.
\begin{lemma}[handling the probability of catching a heavy hitter] \label{lem:large}
Let $\ov{V}$ be the value in Line \ref{lin:ov_v_value} of Algorithm \ref{alg:btree}. Let $V = \epsilon \ov{V}$ be the value in Line \ref{lin:v_value} of Algorithm \ref{alg:btree}.
Let $j' \in  T_{\ell-1}$, and let $j$ be one of its children. Let $z_j$ be defined as follows,
\begin{align*}
z_j = \median_{r \in [R]} \left| y_{\ell,r,h_{\ell,r}(j)} \right|^2.
\end{align*}
 If $\|x_{I_{\ell,j}}\|_2^2 \geq C_j\frac{\epsilon}{k} \|x_{-k}\|_2^2$, where $C_j \geq 2$, then with probability $1-C_j^{-R/6}$ (over the randomness of $h_{\ell,r}$ and $g_{i,\ell,r}$ for $r \in [R], i \in [n]$ in Definition~\ref{def:sketching_matrix}), we have that
\begin{align*}
z_j \geq \eta V .
\end{align*}
\end{lemma}

\begin{proof}

Fix $ r\in [R]$. Let $ b = h_{\ell,r}(j)$, and define $J = \cup_{ t \in h^{-1}_{\ell,r}(b)} I_{\ell,t}$. We observe that
\begin{align*}
|y_{\ell,r,b }|^2 =\|x_J \|_2^2 g^2, \text{~where~} g \sim \mathcal{N}(0,1).
\end{align*}

By property of Gaussian distribution, we have
\begin{align*}
\Pr \left[ |y_{\ell,r,b}|^2 \leq (\eta/C_j) \|x_J\|_2^2 \right] \leq \frac{2}{\sqrt{2 \pi}} \sqrt{ \frac{\eta}{ C_j} } \leq \sqrt{\frac{2\eta}{\pi C_j}},
\end{align*}

It implies, since $\|x_J\|_2^2 \geq \|x_{I_{\ell,j}}\|_2^2 \geq C_j\frac{\epsilon}{k} \|x_{-k}\|_2^2$, that
\begin{align*}
\Pr \left[ |y_{\ell,r,b}|^2 \leq \eta  \frac{\epsilon}{k} \|x_{-k}\|_2^2\right] \leq \sqrt{\frac{2\eta}{\pi C_j}}.
\end{align*}

Since Lemma~\ref{lem:lplp_tail_estimation}, we have $ \ov{V} \leq \frac{1}{ k } \|x_{-k}\|_2^2$.  Because $V =\epsilon \ov{V}$,
\begin{align*}
 \Pr \left[ |y_{\ell,r,b }|^2 \leq \eta V \right] \leq \sqrt{\frac{2\eta}{\pi C_j}}.
\end{align*} 

The $R$ repetitions ensure that the failure probability can be driven down to $C_j^{-R/6}$, because
\begin{align*}
\Pr [z_j \leq \eta V] \leq &~ {R \choose R/2} \cdot \left( \sqrt{\frac{2\eta}{\pi C_j}} \right)^{R/2} \\
\leq &~ 2^R \cdot (2\eta / \pi)^{R/4} \cdot C_j^{-R/4} \\
\leq &~ (2^R \cdot (2\eta / \pi)^{R/4} \cdot 2^{-R/12}) \cdot C_j^{-R/6} \\
\leq &~ \left(2^{11} \cdot (2/9\pi)^3 \right)^{R/12} \cdot C_j^{-R/6} \\
\leq &~ C_j^{-R/6},
\end{align*}

where the first step follows from a union bound, the third step follows from $C_j \geq 2$, and the forth step follows from $\eta \leq 1/9$.

\end{proof}

The following lemma handles the probability of a non-heavy interval being considered ``heavy'' by the algorithm at level $\ell$.

\begin{lemma}[handling the probability of false positives]\label{lem:small}
Let $V$ be the value in Line \ref{lin:v_value} of Algorithm \ref{alg:btree}.
Let $j'$ be an index in $T_{\ell-1}$, and let $j$ be one of its children. If $\|x_{I_{\ell,j}}\|_2^2 \leq \zeta \frac{\epsilon}{k} \|x_{-C_0k}\|_2^2$, then with probability $1 - 2^{-R/3}$ (over the randomness of $h_{\ell,r}$ and $g_{i,\ell,r}$ for $r \in [R], i \in [n]$ in Definition~\ref{def:sketching_matrix}) we have that 
\begin{align*}
z_j < \eta V.
\end{align*}
\end{lemma}

\begin{proof}
Fix $\ell \in [H]$ and consider the set $H_{\ell}$ that contains the $C_0 k$ coordinates $j''$ with the largest $\|x_{I_{\ell,j''}}\|_2^2$ values. Define $H_{\ell}^{(j)} = H_{\ell} \setminus \{j \}$. Fix $r \in [R]$ and observe that by a union-bound we get that
\begin{align*}
\Pr \left[ \exists j'' \in H_{\ell}^{(j)} ~\big|~ h_{\ell,r}(j) = h_{\ell,r}(j'') \right] \leq C_0k \cdot \frac{1}{C_Bk/ \epsilon} = \frac{C_0 \epsilon}{C_B} \leq \frac{1}{20},
\end{align*}
because $C_B \geq 20 C_0$.

We condition on the event $\forall j'' \in H_{\ell}^{(j)}: h_{\ell,r}(j) \neq h_{\ell,r}(j'') $. A standard calculation now shows that
\begin{align*}
\E \left[ |y_{\ell,r,h_{r,\ell}(j)}|^2 \right] 
\leq & ~  \|x_{I_{\ell,j}}\|_2^2 + \frac{1}{C_B} \frac{\epsilon}{k} \|x_{-C_0k}\|_2^2 \\
\leq & ~  \zeta \frac{\epsilon}{k} \| x_{-C_0 k} \|_2^2 + \frac{1}{C_B} \frac{\epsilon}{ k} \| x_{-C_0 k} \|_2^2  \\
\leq & ~ \frac{2\zeta \epsilon}{k} \|x_{-C_0k}\|_2^2,
\end{align*}
where the last step follows from $\frac{1}{C_B} < \zeta$.

We now apply Markov's inequality to obtain
\begin{align*}
     \Pr \left[ | y_{\ell,r,h_{\ell,r}(j)}|^2 \geq \eta V \right ]  
\leq & ~ \Pr \left[| y_{\ell,r,h_{\ell,r}(j)}|^2 \geq \frac{\eta \epsilon}{10 k} \|x_{-C_0k}\|_2^2 \right] \\
\leq & ~ \frac{\frac{2 \zeta \epsilon}{k} \|x_{-C_0k}\|_2^2 }{ \frac{\eta \epsilon}{10 k} \|x_{-C_0k}\|_2^2 } \\
= & ~ \frac{20 \zeta}{ \eta} \\
\leq & ~ \frac{1}{20}, & \text{~by~} \zeta \leq \eta/400.
 \end{align*}

By a union bound, the unconditional probability $\Pr \left[ | y_{\ell,r,h_{\ell,r}(j)}|^2 \geq \eta V \right ] \leq \frac{1}{10}$. 
Finally, we can upper bound the probability that $z_j = \mathrm{median}_{r\in [R]} |y_{\ell,r,h_{\ell,r}(j) }|^2$ is greater than $\eta V$,
\begin{align*}
\Pr [z_j \geq \eta V] \leq & ~  {R \choose R/2} (1/10)^{R/2} \\
< & ~  2^{R} \cdot 2^{-\frac{3}{2} R} \\
= & ~  2^{-\frac{R}{2}} \\
< & ~  2^{-\frac{R}{3}}.
\end{align*}

\end{proof}

\subsection{Analysis of running time}\label{sec:btree_running_time}

\begin{lemma}[bounds of $D,H$ with respect to $R$]\label{lem:bounds}
Let $D,H$ as in Definition \ref{def:degree_height}. It holds that $ D \leq 2^{\frac{R}{6}-10},$ and $ H \leq 2^{\frac{R}{6}-10}$.
\end{lemma}

\begin{proof}

Since $H \geq D$, it suffices to prove the claim only for $H$.

\begin{align*}
H  = C_H \frac{\log q}{ \log \log q}< C_H\log q = C_H\log \left( \epsilon n / k \right) = C_H 2^{\log \log(\epsilon n/k)} \leq 2^{\frac{R}{6} - 10},
\end{align*}
where the third step follows from $q = \epsilon n / k$, and the last step follows from $\log \log (\epsilon n / k) + \log C_H \leq (C_R / 6) \log \log ( \epsilon n /k ) - 10$ and note that $\log\log (\epsilon n / k) \geq 2$ because we assume $k / \epsilon \leq n / 16$.

\end{proof}

\begin{lemma}[running time]\label{lem:running_time}
Let $R$ as in Definition \ref{def:R} and $D$, $H$ as in Definition \ref{def:degree_height}.
Let $T_H$ be the set obtained by applying procedure \textsc{RecursiveBTree} $H$ times (lines \ref{lin:beg_of_btree}-\ref{lin:end_of_btree}) of Algorithm \ref{alg:btree}, and let $C_L > 1$ be some sufficiently large absolute constant. With probability $1- H \cdot 2^{-R/6 + 1}$ we have that: 
\begin{itemize}
\item $|T_{H}| \leq C_L \cdot k/\epsilon$,
\item The running time of \textsc{BTreeSparseRecovery} is 
\begin{align*}
O\left( \epsilon^{-1} k \cdot \log(\epsilon n/ k) \cdot D \right).
\end{align*}
\end{itemize}
\end{lemma}

\begin{proof}

Let $C_L = 2(C_0 + 1 / \zeta)$.

First, it is easy to see that $|T_0|$ is bounded by 
\begin{align*}
|T_0| = \tau = k / \epsilon < C_L k / \epsilon.
\end{align*}

 We claim that if we condition on the event that $|T_{\ell-1}| \leq C_L k/\epsilon$, then with probability $1-2^{-R/6+1}$, $|T_{\ell}| \leq C_L k / \epsilon $. The proof of both bullets  will then follow by a union-bound over all $H$ levels. Indeed, consider the set $Q_{\ell}$ containing the $|T_{\ell-1}| D \leq C_L \cdot (k/\epsilon) \cdot D$ coordinates $j$ that are children of some $j' \in T_{\ell-1}$. Define
\begin{align*}
B_{\ell} = \left\{ j \in Q_{\ell} ~\bigg|~ \|x_{I_{\ell,j}}\|_2^2 \leq \zeta \frac{\epsilon}{k} \|x_{-C_0k}\|_2^2 \right\}.
\end{align*}


By definition of $B_{\ell}$ and $Q_{\ell}$, we have
\begin{align*}
|B_{\ell}| \leq |Q_{\ell}| \leq C_L \cdot (k / \epsilon) \cdot D.
\end{align*}

Moreover, Lemma \ref{lem:small} gives
\begin{align*}
 \forall j \in B_{\ell},  \Pr \left[ z_j \geq \eta V \right] \leq 2^{-R/3}
\end{align*}

Define random variables $W_j$ to be $1$ if $z_j \geq \eta V$, and $0$ otherwise. Then
\begin{align*}
\E \left[ \sum_{j \in B_{\ell}} W_j \right]
\leq & ~ \frac{ C_L k}{ \epsilon} D \cdot 2^{-R/3}
\end{align*}

An application of Markov's inequality gives
\begin{align*}
\Pr \left[ \sum_{j \in B_{\ell}} W_j \geq  \frac{ C_L k}{ 2\epsilon} D \cdot 2^{-R/6} \right] \leq 2^{-R/6+1}.
\end{align*}

Conditioning on $\sum_{j \in B_{\ell}} W_j \leq 
 \frac{ C_L k}{ 2\epsilon} D \cdot 2^{-R/6}$  we will upper bound the size of $T_{\ell}$.

First, observe that there exist at most $(C_0k + k/(\zeta \epsilon))$ $j \in \mathcal{I}_{\ell}$ for which $\|x_{I_{\ell,j}}\|_2^2 > \zeta \frac{\epsilon}{k} \|x_{-C_0 k}\|_2^2$. This gives
\begin{align}\label{eq:btree_size_T_ell}
|T_{\ell}| \leq \left( C_0k + \frac{k}{ \zeta \epsilon} \right) + \frac{C_L k}{2\epsilon}D 2^{-R/6}. 
\end{align}

If $C_L \geq 2 (C_0 + 1/\zeta)$, then we can upper bound the first term in Eq.~\eqref{eq:btree_size_T_ell},
\begin{align*}
 C_0k + \frac{k}{ \zeta \epsilon}  \leq \left( C_0+ \frac{1}{\zeta} \right) \frac{k}{\epsilon} \leq \frac{1}{2} \frac{ C_L k }{ \epsilon }.
\end{align*}

For the second term in Eq.~\eqref{eq:btree_size_T_ell}, we can show that 
\begin{align*}
\frac{C_L k}{2\epsilon} \cdot D 2^{-R/6} \leq \frac{1}{2} \frac{C_L k}{ \epsilon},
\end{align*}

or equivalently $ D \leq 2^{R/6}$, which holds by Lemma \ref{lem:bounds}.

We have $D$ levels, and at each level $\ell$ we have $|T_\ell| = O(k/\epsilon)$, conditioned on the aforementioned events happening. The children of $T_\ell$ is then $O(k/\epsilon \cdot D)$. Since we have $R$ repetitions the total running time per level is $O( (k/\epsilon) \cdot D \cdot R)$, and the total running time is 
\begin{align*}
O( (k/\epsilon) \cdot D \cdot R \cdot H ) = O \left( (k/\epsilon) \cdot D \cdot \log \log q \cdot \frac{\log q}{\log \log q} \right) = O \left( (k/\epsilon) \cdot D \cdot \log( \epsilon n / k ) \right) ,
\end{align*}
where the first step follows from definition of $R$ and $H$, and the last step follows from definition of $q$.

Therefore, it gives the desired result.
\end{proof}

\subsection{Guarantees of the algorithm}\label{sec:btree_guarantee}

\begin{lemma}[guarantees]\label{lem:final}
Let $L = T_H$ be the set obtained by applying procedure \textsc{RecursiveBTree} $R$ times (lines \ref{lin:beg_of_btree}-\ref{lin:end_of_btree}) of Algorithm \ref{alg:btree}, we have that, with probability $9/10$, there exist $T' \subseteq L$ of size at most $k$, such that 
\begin{align*} 
	 \|x - x_{T'}\|_2^2  \leq (1 + \epsilon) \|x_{-k}\|_2^2.
\end{align*}
\end{lemma}

\begin{proof}

Define 
\begin{align*}
\mathcal{H} = \left\{ j \in H(x,k) ~\bigg|~ \exists C_j \geq 2,  |x_j|^2 \geq C_j \frac{\epsilon}{k}\|x_{-k}\|_2^2 \right\}.
\end{align*}

Moreover, associate every $ j \in \mathcal{H}$ with its corresponding $C_j = \frac{|x_j|^2}{ \frac{\epsilon}{k}\|x_{-k}\|_2^2}$.

Pick $j \in \mathcal{H}$. Let also $j_1,j_2,\ldots,j_{H }$, be numbers such that
\begin{align*}
 j \in I_{1,j_1}, j \in I_{2,j_2},\ldots ,j \in I_{R, j_{H}}.
\end{align*} 
For any $t \in \{1,\ldots, H-1 \}$, if $ I_{t,j_t} \in T_t$, then $I_{t+1,j_{t+1}} \in T_{t+1}$ with probability $1 - C_j^{-R/6}$, by Lemma \ref{lem:large}. Since $C_j \geq 2$, this allows us to take a union bound over all $H$ levels and claim that with probablity $1 -HC_j^{-R/6}$, $j \in T_{R}$. For $j \in \mathcal{H}$ define random variable to $\delta_j$ to be $1$ if $ j \notin T_H $. 
\begin{align*}
 \E \left[\delta_j x_j^2 \right]
\leq & ~ H\cdot C_j^{-R/6}x_j^2 \\
\leq & ~ H\cdot C_j^{-R/6}C_j \frac{\epsilon}{k}\|x_{-k}\|_2^2 \\
 = & ~ H \cdot C_j^{-R/6 + 1 } \frac{\epsilon}{k}\|x_{-k}\|_2^2 \\
 \leq & ~ \frac{\epsilon}{80 k} \|x_{-k}\|_2^2,
\end{align*}
where the first step follows by definition of $\delta_j$, the second step follows by the fact that $j \in \mathcal{H}$, and the last step follows by Lemma \ref{lem:bounds}.
Since $|\mathcal{H}| \leq k$, we have that
\begin{align*}
\E \left[ \sum_{j \in \mathcal{H}} \delta_j x_j^2 \right] \leq |{\cal H}| \cdot \frac{ \epsilon }{ 80 k } \| x_{-k} \|_2^2 \leq \frac{\epsilon}{80}\|x_{-k}\|_2^2.
\end{align*}
Then applying Markov's inequality, we have
\begin{align*}
\Pr \left[ \sum_{j \in \mathcal{H} }\delta_j x_j^2 > \frac{\epsilon}{2}\|x_{-k}\|_2^2\right] \leq \frac{1}{40}.
\end{align*}

We condition on the event $\sum_{j \in \mathcal{H}} \delta_j x_j^2 \leq \frac{\epsilon}{2} \|x_{-k}\|_2^2$. Setting $T' = \mathcal{H} \cap T_H$ we observe that
\begin{align*}
 \|x - x_{T'}\|_2^2  = & ~ \sum_{j \in \mathcal{H} \cap H(x,k) } \delta_j x_j^2 + \|x_{ H(x,k) \setminus \mathcal{H}} \|_2^2  + \|x_{[n] \setminus H(x,k)}\|_2^2 \\
 \leq & ~ \frac{\epsilon }{2}  \|x_{-k}\|_2^2 + k\cdot \frac{2\epsilon}{k} \|x_{-k}\|_2^2 + \|x_{-k}\|_2^2 \\
\leq & ~ ( 1+ 3\epsilon) \|x_{-k}\|_2^2.
\end{align*}
where the second step follows by the bound on $\sum_{j \in T } \delta_j x_j^2$ and the fact that every  $j \notin T$ satisfies $ |x_j|^2 \leq (2\epsilon/k)\|x_{-k}\|_2^2$.



Rescaling for $\epsilon$ we get the desired result.
\end{proof}

\subsection{Bounding the number of measurements}\label{sec:btree_measurements}

In this subsection, we prove that 

\begin{claim}[$\#$measurements]
The number of measurements is $O( \epsilon^{-1} k \cdot \log(\epsilon n/k) )$.
\end{claim}
\begin{proof}

Recall the definition of $\tau$ and $q$,
\begin{align*}
\tau = k/\epsilon, ~~~ q = n /\tau.
\end{align*}

We thus have the following bound on the number of measurements: 
\begin{align*}
B \cdot H \cdot R = C_B (k/\epsilon) \cdot C_H \frac{\log (\epsilon n/k)}{\log \log(\epsilon n/k)} \cdot C_R \log \log(\epsilon n/k) = O \left( (k/\epsilon) \log (\epsilon n/k) \right).
\end{align*}

\end{proof}

%% file: pruning.tex
\section{The Pruning Linear Sketch}\label{sec:pruning_sketch}
The goal of this section is to prove Theorem~\ref{thm:pruning}.
\begin{theorem}[Restatement of Lemma~\ref{lem:pruning_sketch}]\label{thm:pruning}
Let $C_L, \alpha, \beta> 1$ be three fixed constants. There exists a randomized construction of a matrix $\Phi \in \mathbb{R}^{m\times n}$, with $m = O( (k/\epsilon ) \cdot \log(1/\epsilon)),$ with column sparsity $O(\log(1/\epsilon))$ such that the following holds :\\
Suppose that one is given a set $L \subseteq [n] $ such that
\begin{align*}
 |L| = C_L \cdot k/\epsilon, && \exists T \subset L, |T| \leq k: \|x - x_T \|_2 \leq ( 1+\epsilon) \|x_{-k}\|_2.
\end{align*}
Then procedure \textsc{Prune} (Algorithm~\ref{alg:prune}) can find a set $S$ of size $\beta \cdot k$  in time $O(m)$, such that
\begin{align*}
\|x - x_{S} \|_2 \leq (1 + \alpha \cdot \epsilon) \|x_{-k}\|_2
\end{align*}
holds with probability $9/10.$
\end{theorem}

In Section~\ref{sec:pruning_sketch_definition}, we provide some basic definitions and description of our algorithm. We analyze the coordinates from several perspectives in Section~\ref{sec:pruning_sketch_analyzing_coordinates}. We prove the correctness of our algorithm in Section~\ref{sec:pruning_sketch_analyzing_correctness} and analyze time, number of measurements column sparsity, success probability of algorithm in Section~\ref{sec:pruning_sketch_analyzing_others}.


\subsection{Design of the sketching matrix, and helpful definitions}\label{sec:pruning_sketch_definition}

\begin{table}[!h]
\centering
\begin{tabular}{|l|l|l|l|}
\hline
  Notation & Choice & Statement & Parameter\\ \hline
  $C_R$ & $10^4+500C_L$  & Definition~\ref{def:Phi_in_pruning} & $R$ \\ \hline
  $C_B$ &  $5 \times 10^5$ & Definition~\ref{def:Phi_in_pruning} & $B$ \\ \hline
  $C_g$ & $4/5$ & Fact~\ref{fact:gaussian_fact} & Gaussian variable \\ \hline
  $C_L$ & $10^4$ & Theorem~\ref{thm:pruning} & $L$ \\ \hline
  $\alpha$ & $5$ & Theorem~\ref{thm:pruning} & Blow up on $\epsilon$ \\ \hline
  $\beta$ & $100$ & Theorem~\ref{thm:pruning} & Blow up on $k$ \\ \hline
\end{tabular}\caption{Summary of constants in Section~\ref{sec:pruning_sketch}, the column ``Parameter'' indicates which parameter is depending on that constant. Note that set $L$ is the input of the algorithm in Section~\ref{sec:pruning_sketch} and the output of the algorithm in Section~\ref{sec:btree}.}
\end{table}

\begin{definition}[sketching matrix $\Phi$]\label{def:Phi_in_pruning}
Let $C_R,C_B > 1$ be absolute constants. Let $R = C_R \log(1/\epsilon)$. Let $B = C_B k / \epsilon$. For $r \in [R]$, we pick $2$-wise independent hash function $h_r: [n] \rightarrow [B]$, as well as normal random variables $\{g_{i,r} \}_{i \in [n], r \in [R]}$ and take measurements
\begin{align*}
y_{r,b} = \sum_{ i \in h^{-1}_r(b)} x_i g_{i,r}.
\end{align*}

\end{definition}

Given the set $L$, for every $i \in L$ we calculate 
\begin{align*}
z_i = \median_{ r \in [R] } |y_{r,h_r(i)}|,
\end{align*}
and keep the indices $i$ with the $\beta k$ largest $z_i$ values to form a set $S$ of indices, for some absolute constant $\beta$ sufficiently large.
We describe this pruning step in Algorithm~\ref{alg:prune}.
For the analysis, we define the threshold 
\begin{align}\label{eq:threshold_tau}
\tau = \|x_{-k}\|_2/\sqrt{k}.
\end{align}
 We will need the following standard fact about the Gaussian distribution. Then we proceed with a series of definitions and lemmata.

\begin{algorithm}[!t]\caption{the prune procedure}\label{alg:prune}
\begin{algorithmic}[1]
\Procedure{\textsc{Prune}}{$x,n,k,\epsilon,L $} \Comment{Theorem~\ref{thm:pruning}} 
  \State $R \leftarrow C_R \log(1/\epsilon)$
	\State $B \leftarrow C_B k / \epsilon$
  \For{$r=1 \to R$}
		\State Sample $h_r: [n] \rightarrow [B]$ $\sim 2$-wise independent family
    \For{$i=1 \to n$}
      \State Sample $g_{i,r} \sim {\cal N}(0,1)$
    \EndFor
  \EndFor
  \For{$r=1 \to R$}
    \For{$b=1 \to B$}
      \State $y_{r,b} \leftarrow \sum_{ i \in h^{-1}_r(b)} x_i g_{i,r}$
    \EndFor
  \EndFor
  \For{$i \in L$}
    \State $z_i \leftarrow \median_{ r \in [R] } |y_{r,h_r(i)}|$
  \EndFor
	\State $S \leftarrow \{i \in L: z_i \text{ is in the top }\beta k\text{ largest~coordinates~in~vector~}z\}$
	\State \Return $S$
\EndProcedure
\end{algorithmic}
\end{algorithm}

\begin{fact}[property of Gaussian]\label{fact:gaussian_fact}
Suppose $x \sim \N(0,\sigma^2)$ is a Gaussian random variable. For any $ t \in (0,\sigma]$ we have
\begin{align*}
\Pr[ x \geq t ] \in \left[ \frac{1}{2} ( 1 - \frac{4}{5} \frac{t}{\sigma} ), \frac{1}{2} ( 1 - \frac{2}{3} \frac{t}{\sigma} ) \right].
\end{align*}
Similarly, if $x \sim \N(\mu, \sigma^2)$, for any $t \in (0,\sigma]$, we have
\begin{align*}
\Pr[ |x| \geq t ] \in \left[ 1 - \frac{4}{5} \frac{t}{\sigma} , 1 - \frac{2}{3} \frac{t}{\sigma} \right].
\end{align*}

The form we will need is the following:
\begin{align*}
	\Pr_{g \sim \mathcal{N}(0,1)}[ |g| \leq t ] \leq \frac{4}{5}t. 
\end{align*}

Thought the analysis, for convenience we will set $C_g = 4/5$. Another form we will need is:
\begin{align*}
\Pr_{g \sim \mathcal{N}(0,1)} \left[ |g| \in \Big[ \frac{1}{3C_g} ,2 \Big] \right] \geq 0.63
\end{align*}
\end{fact}
\begin{proof}
The first form is true by simple calculation. 
The second form is holding due to numerical values of cdf for normal distribution, 
\begin{align*}
\Pr \left[ |g_{i,r}| \in \Big[ \frac{1}{3C_g}, 2 \Big] \right] = 2 (f(2) - f(1/3C_g)) = 2(f(2) - f(5/12)) \geq 2 (0.977 - 0.662) = 0.63,
\end{align*}
 where $f(x) = \int_{-\infty}^{x} \frac{e^{-x^2 / 2}}{\sqrt{2\pi}} dx$ is the cdf of normal distribution. 
\end{proof}

Stochastic dominance is a partial order between random variables and it is a classical concept in decision theory and decision analysis \cite{hr69,b75}. We give the simplest definition below and it is sufficient for our application.
\begin{definition}[stochastic domination of Gaussian random variables]
Let $\sigma_1 < \sigma_2$ and random variables $X \sim \mathcal{N}(0,\sigma_1^2), Y \sim \mathcal{N}(0,\sigma_2^2)$. Then we say that $|Y|$ stochastically dominates $|X|$, and it holds that
\begin{align*}
\Pr \left[ |Y| \geq \lambda \right] \geq \Pr \left[ |X| \geq \lambda \right], ~~~ \forall \lambda \geq 0.
\end{align*}
\end{definition}

We formally define the set $L$ as follows:
\begin{definition}[set $L$, input of the algorithm]\label{def:prune_input_set_L}
Let $C_L > 1$ be a fixed constant, and let set $L \subseteq [n]$ be defined as:
\begin{align*}
| L | = C_L \cdot k/\epsilon, && \exists T \subset |T| \leq k : \| x - x_T \|_2 \leq (1+\epsilon) \| x_{-k} \|_2.
\end{align*}
\end{definition}

We provide a definition called ``badly-estimated coordinate'',
\begin{definition}[badly-estimated coordinate]\label{def:badly-estimated-coordinates}
We will say a coordinate $i \in [n]$ is badly-estimated if
\begin{align*}
 z_i \notin \left[ \frac{1}{3C_g}|x_i| - \frac{1}{100} \frac{\sqrt{\epsilon}}{\sqrt{k}} \|x_{-k}\|_2, 2 |x_i| + \frac{1}{100} \frac{\sqrt{\epsilon}}{\sqrt{k}} \|x_{-k}\|_2 \right],
\end{align*}
\end{definition}

Then, we can define ``badly-estimated set'',
\begin{definition}[badly-estimated set ${\cal B}$]\label{def:badly-estimated-set}
Let set $L$ be defined as Definition~\ref{def:prune_input_set_L}. We say $\mathcal{B} \subseteq L$ is a badly-estimated set if for all $i \in {\cal B}$, $z_i$ is a badly estimated coordinate (see Definition~\ref{def:badly-estimated-coordinates}).
\end{definition}

We define a set of large coordinates in head,
\begin{definition}[large coordinates in head]
Let $\tau$ be defined in \eqref{eq:threshold_tau}. Let $L$ be defined in Definition~\ref{def:prune_input_set_L}. Let $C_g$ be the constant from Fact \ref{fact:gaussian_fact}. Define set
\begin{align*}
 \mathcal{M} = \{ i \in L \cap H(x,k): |x_i| \geq 3C_g \tau \},
\end{align*}
which contains the head coordinates of $x$ that are in $L$ and are larger in magnitude than $3C_g\tau$.
\end{definition}

\subsection{Analyzing head and badly-estimated coordinates}\label{sec:pruning_sketch_analyzing_coordinates}

\begin{lemma}[expected error from coordinates above $\tau$]\label{cla:prune_claim_1}
We have that
\begin{align*}
\E \left[ \sum_{ i \in \mathcal{M} }  x_i^2 \cdot \mathds{1}_{z_i < \tau} \right]\leq \frac{\epsilon}{100}  \|x_{-k}\|_2^2.
\end{align*}
\end{lemma}

\begin{proof}

Fix $ i \in \mathcal{M}$. Observe that for $r \in [R]$
\begin{align*}	
|y'_{r,h_r(i)}| \sim  \|x_{h^{-1}_r(i)}\|_2 |\mathcal{N}(0,1)|.
\end{align*}

Since 
\begin{align*}
\|x_{h^{-1}_r(i)}\|_2 \geq |x_i|
\end{align*}
we have that the random variable $|y'_{r,h_r(i)}|$ stochastically dominates the random variable $|x_i| \cdot |\mathcal{N}(0,1) |$. 

By Fact \ref{fact:gaussian_fact}, we have that 
\begin{align*}
\Pr \left[ |y'_{r,h_r(i)}| \leq \tau \right] \leq C_g \frac{\tau}{|x_i|} .
\end{align*}

Because of the $R=C_R \log(1/\epsilon)$ repetitions, a standard argument gives that 
\begin{align*}
\Pr \left[ \mathds{1}_{z_i < \tau} = 1 \right ]	\leq  \left( C_g \frac{\tau}{|x_i|} \right)^{ C' \log(1/\epsilon)},
\end{align*}
for some absolute constant $C' > C_R /3$.

We now bound 	
\begin{align*}	
  \E \left[ \sum_{i \in \mathcal{M} } x_i^2 \cdot \mathds{1}_{z_i < \tau} \right] \leq & ~ \sum_{i \in \mathcal{M}} x_i^2\left( C_g \frac{\tau}{|x_i|} \right)^{ C'\log(1/\epsilon)} \\
\\ = & ~ \sum_{i \in \mathcal{M}} C_g^2 \tau^2 \left( C_g \frac{\tau}{|x_i|} \right)^{C' \log(1/\epsilon) -2 } \\ 
 \leq & ~ k \cdot \tau^2 \cdot \frac{\epsilon }{100} \\
  = & ~ \frac{\epsilon}{100}\|x_{-k}\|_2^2,
\end{align*}

where the first step follows by the bound on $ \E \left[ \mathds{1}_{z_i < \tau} \right] = \Pr  \left[ \mathds{1}_{z_i < \tau} = 1  \right]$, and the third step by choosing by choosing $C_R>1$ to be some sufficiently large constant and the facts $C'> C_R/3$ and $(C_g \tau)/|x_i| \leq 1/3$. 

\end{proof}

\begin{lemma}[probability of a fixed coordinate is badly-estimated]\label{cla:prune_claim_2}
A coordinate $i$ is badly-estimated (as in Definition~\ref{def:badly-estimated-coordinates}) probability at most 
\begin{align*}
\frac{ \epsilon^3 }{ 100^2C_L }.
\end{align*}
\end{lemma}

\begin{proof}

Fix $r$ and set set $b = h_r(i)$. Recall the definition of $y_{r,b} = \sum_{ i \in h^{-1}_r(b)} x_i g_{i,r}$ in Definition~\ref{def:Phi_in_pruning}. We have that 
\begin{align*}
 \left| |g_{i,r} x_i | - \left| \sum_{j \in h_r^{-1}(b)\setminus \{i\}} g_{j,r} x_j \right| \right| \leq |y_{r,b} |
\leq  |g_{i,r} x_i | + \left| \sum_{j \in h_r^{-1}(b)\setminus \{i\}} g_{j,r} x_j \right|,
\end{align*}
Now, $|g_{i,r}x_i|$ will be at in [$(1/3C_g) |x_i|, 2|x_i|]$ with probability at least $0.63$ (due to Fact~\ref{fact:gaussian_fact}).  

Moreover, for any $j \in H(x,k)\setminus\{i\}$, $h_r(j) \neq b$ with probability $1 - 1/B = 1-\epsilon/(C_Bk) \geq 1-1/(C_Bk)$.  By a union bound, we get with probability at least $1-1/C_B$, for all $j \in H(x,k) \setminus \{i\}$, $h_r(j) \neq b$. Conditioning on this event, we have,

\begin{align*}
\E \left[ \left(\sum_{j \in h_r^{-1}(b)\setminus \{i\}} g_{j,r} x_j  \right) ^2  \right] = \frac{\epsilon}{C_B k} \|x_{-k}\|_2^2.
\end{align*}

We then apply Markov's inequality to get that with probability at least $1-10^4/C_B$, 
\begin{align*}
\left( \sum_{j \in h_r^{-1}(b)\setminus \{i\}} g_{j,r} x_j \right)^2 \leq  \frac{\epsilon}{10^4 k} \|x_{-k}\|_2^2 .
\end{align*}

Therefore, by a union bound, 
\begin{align*}
\Pr \left[ |y_{r,b}| \in \left[ \frac{ 1 }{ 3C_g } |x_i|- \frac{1}{100} \frac{ \sqrt{\epsilon} }{\sqrt{k} }, 2|x_i|+ \frac{1}{100} \frac{ \sqrt{\epsilon} }{\sqrt{k} } \right] \right]  
\geq & ~ 0.63 - \frac{1}{C_B} - \frac{10^4}{C_B} \\
\geq & ~ 0.6,
\end{align*}

where the last step follows by $C_B \geq 5 \times 10^5$.

Note that $z_i$ is obtained by taking median of $R$ copies of i.i.d. $|y_{r,b}|$. For each $r \in [R]$, we define $Z_r = 1$ if $|y_{r,b}|$ falls into that region, and $0$ otherwise. We have $\E[ \sum_{r=1}^R Z_r ] \geq 0.6 R$. Using Chernoff bound, we have
\begin{align*}
\Pr \left[ \sum_{r=1}^R Z_r < 0.9 \cdot 0.6 R \right] \leq & ~ \Pr \left[ \sum_{r=1}^R Z_r < 0.9 \E[ \sum_{r=1}^R Z_r ] \right] \\ 
\leq & ~ e^{ - \frac{1}{3} 0.1^2 \E[\sum_{r=1}^R Z_r ] } \\
\leq & ~ e^{ - \frac{1}{3} 0.1^2 \cdot 0.6 R }
\end{align*}
Thus,
\begin{align*}
\Pr \left[ z_i \not\in \left[ \frac{1}{3C_g} |x_i|- \frac{1}{100} \frac{ \sqrt{\epsilon} }{\sqrt{k} }, 2|x_i| + \frac{1}{100} \frac{ \sqrt{\epsilon} } {\sqrt{k} } \right] \right]  
\leq & ~ e^{ - \frac{1}{3} 0.1^2 \cdot 0.6 R } \\
\leq & ~ 2^{-0.002 R} \\
= & ~ 2^{-0.002 C_R \log (1/\epsilon)} \\
\leq & ~ 2^{-0.002 (10000 + 500 C_L) \log (1/\epsilon)} \\
\leq & ~ \frac{\epsilon^3}{100^2C_L} ,
\end{align*}
where the third step follows from choice of $C_R$.

\end{proof}

\subsection{Guarantees of the algorithm}\label{sec:pruning_sketch_analyzing_correctness} 

We now proceed with the proof of Theorem \ref{thm:pruning}.
\begin{proof}

By Lemma~\ref{cla:prune_claim_1} and an application of Markov's inequality we have that 
\begin{align*}
 \sum_{ i \in \mathcal{M} } x_i^2 \cdot \mathds{1}_{z_i < \tau} \leq \epsilon \|x_{-k}\|_2^2,
\end{align*}
with probability $99/100$. Let this event be $\mathcal{E}_1$.

Moreover, by Lemma~\ref{cla:prune_claim_2},
\begin{align*}
 \E[ | {\cal B} | ] \leq \frac{ \epsilon^3 |L| }{ 100^2 C_L },
\end{align*}
so, by Markov's inequality, we have
\begin{align*}
\Pr \left[ | {\cal B} | \leq \frac{ \epsilon^2 |L| }{ 100 C_L } \right] \geq 1 - \epsilon/100 \geq 99/100
\end{align*} 
Let this event be $\mathcal{E}_2$.

By taking a union bound, $\mathcal{E}_1$ and $\mathcal{E}_2$ both hold with probability $98/100$.  Plugging size of $|L|$ ($\leq C_L \cdot k/\epsilon$) into equation of event ${\cal E}_2$, we get
\begin{align}\label{eq:upper_bound_size_of_set_B}
| {\cal B} | \leq \frac{ \epsilon^2 |L| }{ 100 C_L } \leq \frac{\epsilon k}{100}.
\end{align}
It means there are at most $\epsilon k /100$ coordinates that badly-estimated.

We remind that our goal is to bound 
\begin{align*}
\|x - x_S\|_2^2 = & ~ \|x_{\bar{S}}\|_2^2 \\
= & ~ \|x_{ \bar{S} \cap {\cal M} }\|_2^2 + \|x_{ \bar{S} \backslash {\cal M} }\|_2^2 \\
= & ~ \|x_{\bar{S} \cap {\cal M} }\|_2^2 + \|x_{ (\bar{S} \setminus \mathcal{M}) \cap \mathcal{B}}\|_2^2 + \|x_{ (\bar{S} \setminus \mathcal{M}) \backslash \mathcal{B} } \|_2^2
\end{align*}

\paragraph{1. Bounding $\|x_{\mathcal{M} \cap \bar{S}}\|_2^2$.}

Consider the set 
\begin{align*}
I = \{ i \in L \setminus \mathcal{M}: |x_i| \geq \tau/3 \},
\end{align*}
which contains the coordinates in $L$ with magnitude in the range $[\frac{1}{3}\tau, 3C_g\tau)$. By the definition of $\tau$, 
clearly, $|I| \leq 3k + k = 4k$, because we can have at most $k$ such elements in $H(x,k)$, and at most $3k$ such elements in the tail $[n] \setminus H(x,k)$. Since the number of badly estimated coordinates is at most $\epsilon k/100$ and the size of $S$ is $\beta k $ for sufficiently large $\beta$, we can have at most $4k + \epsilon k/100 < \beta k$ coordinates $i \in L$ which are not in $\mathcal{M}$ and are larger than $\tau$. This means that all coordinates in $\mathcal{M}$ with estimate $z_i \geq \tau$ will belong to $S$. This implies that 
\begin{align*}
 \mathcal{M} \cap \bar{S}  = \{ i \in \mathcal{M}: z_i < \tau \}, 
\end{align*}
and hence
\begin{align*}
\|x_{\mathcal{M} \cap \bar{S}}\|_2^2 = \sum_{i \in \mathcal{M}} x_i^2 \cdot \mathds{1}_{z_i < \tau} \leq \epsilon \|x_{-k}\|_2^2,
\end{align*}
since we conditioned on event $\mathcal{E}_1$.

\paragraph{2. Bounding $\|x_{ (\bar{S} \setminus \mathcal{M}) \cap \mathcal{B}}\|_2^2$. }

For every $i \in (\bar{S} \setminus \mathcal{M}) \cap \mathcal{B})$ we have the trivial bound $|x_i| \leq \tau$. Since 
$(\bar{S} \setminus \mathcal{M}) \cap \mathcal{B}) \subseteq \mathcal{B}$, because the event $\mathcal{E}_2$ we get that 
\begin{align*}
\|x_{ (\bar{S} \setminus \mathcal{M}) \cap \mathcal{B}}\|_2^2 \leq |\mathcal{B}| \cdot \tau^2 \leq \frac{\epsilon k}{100} \cdot \frac{\|x_{-k}\|_2^2}{k} = \frac{\epsilon}{100} \|x_{-k}\|_2^2,
\end{align*}
where the second step follows from \eqref{eq:upper_bound_size_of_set_B} and \eqref{eq:threshold_tau}. 

\paragraph{3. Bounding $\|x_{ (\bar{S} \setminus \mathcal{M}) \setminus \mathcal{B}}\|_2^2$.}

Observe that set $(\bar{S} \setminus \mathcal{M}) \setminus \mathcal{B}$ consists of well-estimated coordinates that are less than $\tau$	 in magnitude, and their estimates do not belong to the largest $\beta k$ estimates. For convenience, set $ Q = (\bar{S} \setminus \mathcal{M}) \setminus \mathcal{B}$, then it is obvious that $Q = \ov{S} \backslash ( {\cal M} \cup {\cal B} )$.

We define three sets $H_1$, $H_2$, $H_3$ as follows,
\begin{align*}
H_1 = Q \cap T ,
~~~
H_2 = Q \cap \bar{T}, 
~~~
\text{and~~~}
H_3 =  ( \ov{T} \backslash ( {\cal M} \cup {\cal B} ) ) \backslash \ov{S}  .
\end{align*}

Using the definition of $Q$, we can rewrite $H_1$, $H_2$, and $H_3$ as follows
\begin{align*}
H_1 = & ~ ( \ov{S} \backslash ( {\cal M} \cup {\cal B} ) ) \cap T = \ov{S} \cap \ov{ {\cal M} \cup {\cal B} } \cap T, \\
H_2 = & ~ ( \ov{S} \backslash ( {\cal M} \cup {\cal B} ) ) \cap \ov{T} =  \ov{S} \cap \ov{ {\cal M} \cup {\cal B} } \cap \ov{T}, \\
H_3 = & ~ ( \ov{T} \backslash ( {\cal M} \cup {\cal B} ) ) \backslash \ov{S} = S \cap \ov{ {\cal M} \cup {\cal B} } \cap \ov{T}.
\end{align*}

We can show that
\begin{align}\label{eq:property_of_set_H_2_cap}
H_2 \cap H_3
= & ~ (\ov{S} \cap \ov{ {\cal M} \cup {\cal B} } \cap \ov{T} ) \cup ( S \cap \ov{ {\cal M} \cup {\cal B} } \cap \ov{T} ) \notag \\
= & ~ \emptyset,
\end{align}
and 
\begin{align}\label{eq:property_of_set_H_2_cup}
H_2 \cup H_3
= & ~ (\ov{S} \cap \ov{ {\cal M} \cup {\cal B} } \cap \ov{T} ) \cap ( S \cap \ov{ {\cal M} \cup {\cal B} } \cap \ov{T} ) \notag \\
= & ~ \ov{ {\cal M} \cup {\cal B} } \cap \ov{T} \notag \\
= & ~ \ov{T} \backslash ( {\cal M} \cup {\cal B} ).
\end{align}

Then, 
\begin{align*}
\|x_Q\|_2^2 = & ~ \|x_{H_1}\|_2^2 + \|x_{H_2}\|_2^2 \\
= & ~ \|x_{H_1}\|_2^2 + (\|x_{\bar{T} \setminus ( \mathcal{M} \cup\mathcal{B}) }\|_2^2 - \|x_{ H_3 }\|_2^2) \\
\leq & ~ \|x_{H_1}\|_2^2 + \|x_{\ov{T}}\|_2^2 - \|x_{ H_3 }\|_2^2 \\
\leq & ~ \|x_{H_1}\|_2^2 + (1+\epsilon)\|x_{-k}\|_2^2 - \|x_{ H_3 }\|_2^2 , 
\end{align*}
where first step follows from $H_1 \cap H_2 = \emptyset$ and $H_1 \cup H_2 = Q$, the second step follows from Eq.~\eqref{eq:property_of_set_H_2_cap} and \eqref{eq:property_of_set_H_2_cup}, the third step follows from $\|x_{\bar{T} \setminus ( \mathcal{M} \cup\mathcal{B}) }\|_2^2 \leq \| x_{ \ov{T} } \|_2^2$  and the last step follows from $ \| x_{ \ov{T} } \|_2^2 \leq (1+\epsilon) \| x_{-k} \|_2^2$.

We define $d,E,a,b$ as follows
\begin{align}\label{eq:def_d_E_a_b}
 d = |H_1|, &&
 E = \frac{1}{4}\sqrt{\epsilon/k} \|x_{-k}\|_2, &&
 a = \max_{ i \in H_1 } |x_i|, &&
 b = \min_{ i \in H_3 } |x_i|.
\end{align}

Let $i^*$ and $j^*$ be defined as follows:
\begin{align}\label{eq:i_*_and_j_*}
i^* = \arg\max_{i \in H_1 } |x_i|, \text{~~~and~~~}
j^*= \arg\min_{ j \in H_3 } |x_j|.
\end{align}
Recall the definitions of $H_1$ and $H_3$, we know $H_3$ is a subset of $S$ and $H_1$ is a subset of $\ov{S}$. Since the set $S$ contains the largest $\beta k$ coordinates, thus we have
\begin{align*}
z_{j} \geq z_{i} , \forall i \in H_1, j \in H_3 .
\end{align*}
It further implies $z_{j^*} \geq z_{i^*}$.

By Definition~\ref{def:badly-estimated-coordinates}, we have
\begin{align}\label{eq:connect_z_i_x_i}
z_{i^*} \geq \frac{1}{3 C_g} | x_{i^*} | - \frac{1}{100} \frac{ \sqrt{\epsilon} }{ \sqrt{k} } \| x_{-k} \|_2 ,
\end{align}
and
\begin{align}\label{eq:connect_z_j_x_j}
z_{j^*} \leq 2 | x_{j^*} | + \frac{1}{100} \frac{ \sqrt{\epsilon} }{ \sqrt{k} } \| x_{-k} \|_2 .
\end{align}

Then, we can show that $a \leq 6 C_g b + E$ in the following sense:
\begin{align*}
a = & ~ | x_{i^*} | & \text{~by~def.~of~}i^*, a, \eqref{eq:i_*_and_j_*}, \eqref{eq:def_d_E_a_b} \\
\leq & ~ 3 C_g z_{i^*} + \frac{3 C_g}{100} \sqrt{\epsilon / k} \| x_{-k} \|_2  & \text{~by~} \eqref{eq:connect_z_i_x_i} \\
\leq & ~ 3 C_g z_{j^*} + \frac{3 C_g}{200} \sqrt{\epsilon / k} \| x_{-k} \|_2  & \text{~by~} z_{i^*} \leq z_{j^*} \\
\leq & ~ 6 C_g | x_{j^*} | + \frac{6 C_g}{200} \sqrt{\epsilon / k} \| x_{-k} \|_2 & \text{~by~} \eqref{eq:connect_z_j_x_j} \\
= & ~ 6 C_g b + \frac{6 C_g}{200} \sqrt{\epsilon / k} \| x_{-k} \|_2 & \text{~by~def.~of~}j^*, b, \eqref{eq:i_*_and_j_*}, \eqref{eq:def_d_E_a_b} \\
\leq & ~ 6 C_g b + E & \text{~by~def.~of~}E, \eqref{eq:def_d_E_a_b}
\end{align*}


Note that $H_3 = (\bar{T} \setminus ( \mathcal{M} \cup \mathcal{B} ) ) \setminus \ov{S} = S \setminus (T \cup \mathcal{M} \cup \mathcal{B})$. 
Therefore, 
\begin{align*}
|H_3| \geq & ~ |S| - |T| - |\mathcal{M}| - |\mathcal{B}| \\
\geq & ~ \beta k - k - k - k \\
= & ~ (\beta - 3) k.
\end{align*}

Finally, we can have
\begin{align*}
\|x_{H_1}\|_2^2 -  \|x_{H_3}\|_2^2
\leq & ~ da^2 - (\beta - 3) k b^2 \\ 
 \leq & ~ d(6C_gb + E)^2 - (\beta-3) k b^2 & \text{~by~} a \leq 6 C_g b + E \\
 = & ~ (36C_g^2d - (\beta-3) k )b^2 + 12 C_g bd E + d E^2 \\ 
  \leq & ~  (36C_g^2 k - (\beta-3) k )b^2 + 12 C_g b k E + k E^2 & \text{~by~} d \leq k \\
   \leq & ~  (36C_g^2 k - (\beta-5 C_g^2) k )b^2 + 12 C_g b k E + k E^2 & \text{~by~} C_g \geq 4/5 \\
 \leq & ~ -36k C_g^2 b^2 + 12 C_g b k E + k E^2 & \text{~by~} \beta \geq 77 C_g^2 \\ 
 = & ~ - k( 6C_g b - E)^2 + 2kE^2  \\
 \leq & ~ 2k E^2 \\
 \leq & ~ \epsilon \|x_{-k}\|_2^2.
\end{align*}
where the last step follows from definition of $E$.

Thus, we have
\begin{align*}
\| x_Q \|_2^2 \leq (1+2\epsilon) \| x_{-k} \|_2^2.
\end{align*}
{\bf Putting it all together.} We have
\begin{align*}
\| x - x_S \|_2^2 = & ~ \| x_{ \ov{S} \cap {\cal M} } \|_2^2 + \| x_{ ( \ov{S} \backslash {\cal M} ) \cap {\cal B} } \|_2^2 + \| x_{ ( \ov{S} \backslash {\cal M} ) \backslash {\cal B} } \|_2^2 \\
\leq & ~ \epsilon \| x_{-k} \|_2^2 + \frac{\epsilon}{100} \| x_{-k} \|_2^2 + (1+2\epsilon) \| x_{-k} \|_2^2 \\
\leq & ~ (1+4 \epsilon) \| x_{-k} \|_2^2
\end{align*}
Finally, we can conclude $\alpha =5$ and $\beta = 100$.

\end{proof}

\subsection{Time, measurements, column sparsity, and probability}\label{sec:pruning_sketch_analyzing_others}

In this section, we will bound the decoding time, the number of measurements, column sparsity and success probability of algorithm.

\paragraph{Decoding time.} 
For each $i \in L$, we compute $z_i$ to be the median of $R$ values. For this part, we spend $O(|L| \cdot R) = O((k/\epsilon) \cdot \log (1/\epsilon))$ time.
Moreover, calculating the top $\beta k$ estimates in $L$ only takes $O(|L|)$ time.
Therefore, the decoding time is $O((k/\epsilon) \cdot \log (1/\epsilon))$.

\paragraph{The number of measurements.}
The number of measurements is the bucket size $B$ times the number of repetitions $R$, which is $O(BR) = O((k/\epsilon) \cdot \log (1/\epsilon))$.

\paragraph{Column sparsity.}
Each $i \in [n]$ goes to one bucket for each hash function, and we repeat $R$ times, so the column sparsity is $O(R) = O(\log (1/\epsilon))$.

\paragraph{Success probability.}
By analysis in Section~\ref{sec:pruning_sketch_analyzing_correctness}, the success probability is at least $0.98$.

%% file: tail.tex
\section{Tail Estimation}\label{sec:tail_estimation}

In Section~\ref{sec:tail_random_walks}, we present a standard result on random walks. In Section~\ref{sec:tail_p_stable_distribution}, we present some results on $p$-stable distribution. In what follows we asssume that $0<p\leq 2$. We show an algorithm for $\ell_p$ tail estimation in Section~\ref{sec:tail_tail_estimation_in_the_ell_p_norm}.

\subsection{Random walks}\label{sec:tail_random_walks}
\begin{theorem}\label{lem:random_walk}
We consider the following random walk. We go right if $B_i=1$ and we go left if $B_i=0$. The probability of $B_i = 1$ is at least $9/10$ and the probability of $B_i=0$ is at most $1/10$. With at least some constant probability bounded away from $\frac{1}{2}$, for all the possible length of the random walk, it will never return to the origin.
\end{theorem}
This is a standard claim, that can be proved in numerous ways, such as martingales etc. For the completeness, we still provide a folklore proof here.
\begin{proof}

Let $p > 1/2$ be the probability of stepping to the right, and let $q=1-p$. For integer $m \geq 1$, let $P_m$ be the probability of first hitting $0$ in exactly $m$ steps. It is obvious that $P_m = 0$ if $n$ is even, and $P_1 = q$. In order to hit $0$ for the first time on the third step you must Right-Left-Left, so $P_3=pq^2$. To hit $0$ for the first time in exactly $2k+1$ steps, you must go right $k$ times and left $k+1$ times, your last step must be to the left, and through the first $2k$ steps you must always have made at least many right steps as left steps. It is well known that the number o such path is $C_k$, which is the $k$-th Catalan number. Thus,
\begin{align*}
P_{2k+1} = C_k q^k q^{k+1} = C_k \cdot q (pq)^k = \frac{q (pq)^k }{k+1} { 2k \choose k},
\end{align*}
since 
\begin{align*}
C_k = \frac{1}{k+1} {2k \choose k}
\end{align*}
By \cite{w05}, the generating function for the Catalan numbers is
\begin{align*}
c(x) = \sum_{k \geq 0} C_k x^k = \frac{ 1 - \sqrt{1-4 x} }{2x},
\end{align*}
so the probability that the random walk will hit $0$ is
\begin{align*}
\sum_{k \geq 0} P_{2k+1} 
= & ~ q \sum_{k \geq 0} C_k (pq)^k \\
= & ~ q \cdot c(pq) \\
= & ~ q \cdot \frac{ 1 - \sqrt{1-4pq} }{2pq} & \text{~by~definition~of~}c(x) \\
= & ~ \frac{ 1 - \sqrt{ 1 - 4 q (1-q) } }{2p} \\
= & ~ \frac{ 1 - \sqrt{ 1 - 4 q + 4 q^2 } }{ 2p } \\
= & ~ \frac{ 1 - ( 1 - 2q ) }{2p} \\
= & ~ q/p \\
\leq & ~ 1/9.
\end{align*}
Thus, we complete the proof.
\end{proof}

\subsection{$p$-stable distributions}\label{sec:tail_p_stable_distribution}

\begin{figure}[!t]
  \centering
    \includegraphics[width=0.8\textwidth]{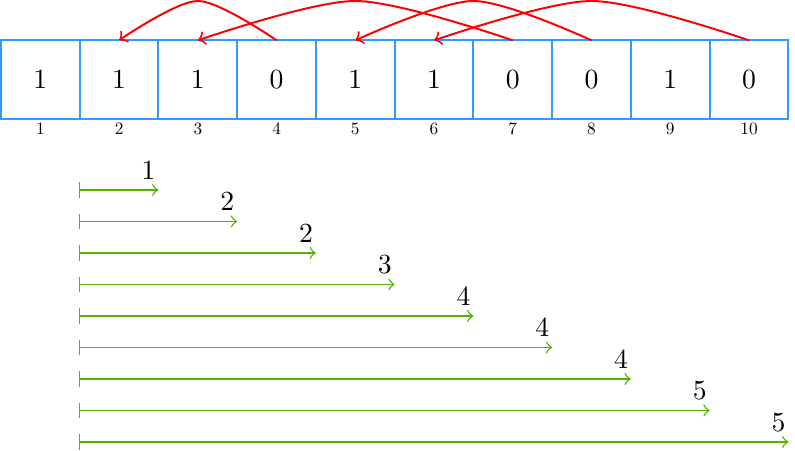}
    \caption{This is a visualization of part of the proof in Claim~\ref{cla:tail_upper_bound_on_Delta_t}. We consider an example where there are $l=10$ blocks, $B_1 = 1$, $B_2 = 1$, $B_1 = 1$, $B_3 = 0$, $B_4=0$, $B_5=1$, $B_6=1$, $B_7=0$, $B_8=0$, $B_9=1$ and $B_{10}=0$. Recall the two important conditions in the proof of Claim~\ref{cla:tail_upper_bound_on_Delta_t}, the first one is $B_1=1$ and the second one is, for all $j\in [l], \sum_{j'=2}^j B_{j'} > (j-1)/2$. The number on the green arrow is $\sum_{j'=2}^j B_{j'}$. It is to see that the example we provided here is satisfying those two conditions. Recall the definition of set $S_1$ and $S_0$. Here $S_1 = \{ 2,3,5,6,9 \}$ and $S_0 = \{ 4,7,8,10 \}$. Then $S'_1=\{ 2,3,5,6 \}$. The mapping $\pi$ satisfies that $\pi(4) = 2$, $\pi(7)=3$, $\pi(8)=5$ and $\pi(10)=6$. }\label{fig:random_walk}
\end{figure}

We first provide the definition of $p$-stable distribution. For the more details, we refer the readers to \cite{i06}.

\begin{definition}[$p$-stable distribution]
A distribution ${\cal D}$ over $\R$ is called $p$-stable, if there exists $p\geq 0$ such that for any $n$ real numbers $a_1, a_2, \cdots, a_n$ and i.i.d. variables $x_1, x_2, \cdots, x_n$ from distribution ${\cal D}$, the random variable $\sum_{i=1}^n a_i x_i$ has the same distribution as the variable $\| a \|_p y$, where $y$ is a random variable from distribution ${\cal D}$. 
\end{definition}

\begin{theorem}[\cite{z86}]
For any $p \in (0,2]$, there exists a $p$-stable distribution.
\end{theorem}
Gaussian distribution defined by the density function $f(x) = \frac{1}{\sqrt{2\pi}} e^{-x^2/2}$, is $2$-stable. Cauchy distribution defined by density function $f(x) = \frac{1}{\pi} \frac{1}{1+x^2}$ is $1$-stable. Let ${\cal D}_p$ denote the $p$-stable distribution. For $p=2$, ${\cal D}_p$ is ${\cal N}(0,1)$ and for $p=1$, ${\cal D}_p$ is ${\cal C}(0,1)$.

\subsection{$\ell_p$-tail estimation algorithm}\label{sec:tail_tail_estimation_in_the_ell_p_norm}
The goal of this Section is prove Lemma~\ref{lem:lplp_tail_estimation}.

\begin{algorithm}
\begin{algorithmic}[1]\caption{$\ell_p$ tail estimation algorithm}\label{alg:lplp_tail_estimation}
\Procedure{\textsc{LpLpTailEstimation}}{$x,k,p,C_0,\delta$} \Comment{Lemma~\ref{lem:lplp_tail_estimation} }
	\State \Comment{Requires $C_0 \geq 1000$}
	\State $m \leftarrow O(\log (1/\delta))$
	\State Choose $g_{i,t}$ to be random variable that sampled i.i.d. from distribution ${\cal D}_p$, $\forall i, t \in [n] \times [m]$
	\State Choose $\delta_{i,t}$ to be Bernoulli random variable with $\E[\delta_{i,t}] = 1/(100 k )$, $\forall i,t \in [n] \times [m]$
	\State \Comment{Matrix $A$ is implicitly constructed based on $g_{i,t}$ and $\delta_{i,t}$}
	\For{$t \in [m]$}
		\State $y_t \leftarrow \sum_{i=1}^n \delta_{i,t} \cdot g_{i,t} \cdot x_i$
	\EndFor
	\State $V \leftarrow \median_{t \in [m]} |y_t|^2$
	\State \Return $V$ \Comment{$ \frac{1}{10k} \| x_{-C_0 k} \|_p^p \leq V \leq \frac{1}{k} \| x_{-k} \|_p^p $}
\EndProcedure
\end{algorithmic}
\end{algorithm}

One can try to prove such a claim for $p=2$ with random signs, instead of Gaussians, by applying the Paley-Zygmund inequality to obtain the lower bound. A straightforward calculation indicates that this approach does not give the desired result, hence we need a new argument to deal with the lower bound.

\begin{lemma}[Restatement of Lemma~\ref{lem:tail_estimation}]\label{lem:lplp_tail_estimation}
Let $C_0 \geq 1000$ denote some fixed constant. There is an oblivious construction of matrix $A\in \R^{m \times n}$ with $m = O(\log (1/\delta))$ along with a decoding procedure $\textsc{LpLpTailEstimation}(x,k,p,C_0,\delta)$ (Algorithm~\ref{alg:lplp_tail_estimation}) such that, given $Ax$, it is possible to output a value $V$ in time $O(m)$ such that 
\begin{align*}
 \frac{1}{10k} \|x_{-C_0 k }\|_p^p \leq V \leq \frac{1}{k} \|x_{-k}\|_p^p,
\end{align*}
holds with probability $1- \delta$.
\end{lemma}

\begin{proof}
Let $m = O(\log (1/\delta))$. For each $i\in [n], t \in [m]$, we use $g_{i,t}$ to denote a random variable that sample from distribution ${\cal D}_p$. 

For each $i\in [n], t \in [m]$, we use $\delta_{i,t} $ to denote a Bernoulli random variable such that
\begin{align*}
\delta_{i,t} =
\begin{cases}
1, & \text{~with~prob.~} \frac{1}{100 k}; \\
0, & \text{~otherwise}.
\end{cases}
\end{align*}
Then we have 
\begin{align*}
\E[ \delta_{i,t} ] = \frac{1}{100 k}.
\end{align*}

For each $t \in [m]$, we define $y_t$ as follows
\begin{align}\label{eq:tail_def_y_t}
y_t  = \sum_{i=1}^n \delta_{i,t} g_{i,t} x_i. 
\end{align}
For each $t \in [m]$, we define $\Delta_t$ as follows
\begin{align}\label{eq:tail_def_Delta_t}
\Delta_t = \left( \sum_{i=1}^n \delta_{i,t}^p x_i^p \right)^{1/p}.
\end{align}

Using Claim~\ref{cla:tail_upper_bound_on_y_t} and Claim~\ref{cla:tail_upper_bound_on_Delta_t}
\begin{align*}
\Pr_{g,\delta} \left[ |y_t| < \alpha \frac{ 1 }{ (2 C_0 k)^{1/p} } \| x_{-C_0 k} \|_p \right] \leq 1/5.
\end{align*}

Using Claim~\ref{cla:tail_lower_bound_on_y_t} and Claim~\ref{cla:tail_lower_bound_on_Delta_t}
\begin{align*}
\Pr_{g,\delta} \left[ |y_t| > \beta \frac{ 1 }{ k^{1/p} } \| x_{-k} \|_p \right] \leq 1/5.
\end{align*}

Finally, we just take the median over $m$ different independent repeats. Since $ m = O(\log (1/\delta))$, thus, we can boost the failure probability to $\delta$.

\end{proof}

It is a standard fact, due to $p$-stability, that $y_t$ follows the $p$-stable distribution : $\Delta_t \cdot {\cal D}_p$. Since $p$-stable distributions are continuous functions,  we have the following two Claims:
\begin{claim}[upper bound on $|y_t|$]\label{cla:tail_upper_bound_on_y_t}
Let $y_t$ be defined in Eq.~\eqref{eq:tail_def_y_t}, let $\Delta_t$ be defined in Eq.~\eqref{eq:tail_def_Delta_t}. There is some sufficiently small constant $\alpha \in (0,1)$ such that
\begin{align*}
\Pr_g [  | y_t | < \alpha \cdot \Delta_t ] \leq 1/10.
\end{align*}
\end{claim}

\begin{claim}[lower bound on $|y_t|$]\label{cla:tail_lower_bound_on_y_t}
Let $y_t$ be defined in Eq.~\eqref{eq:tail_def_y_t}, let $\Delta_t$ be defined in Eq.~\eqref{eq:tail_def_Delta_t}. There is some sufficiently large constant $\beta > 1$ such that
\begin{align*}
\Pr_g [ | y_t | > \beta \cdot \Delta_t ] \leq 1/10.
\end{align*}
\end{claim}

It remains to prove Claim~\ref{cla:tail_lower_bound_on_Delta_t} and Claim~\ref{cla:tail_upper_bound_on_Delta_t}.

\begin{claim}[lower bound on $\Delta_t$]\label{cla:tail_lower_bound_on_Delta_t}
Let $\Delta_t$ be defined in Eq.~\eqref{eq:tail_def_Delta_t}. Then we have 
\begin{align*}
\Pr_{\delta} \left[ \Delta_t > \frac{1}{ k^{1/p} } \| x_{-k} \|_p \right] \leq 1/10.
\end{align*}
\end{claim}
\begin{proof}
The proof mainly includes three steps,

First, for a fixed coordinate $i\in [n]$, with probability at most $1/(100k)$, it got sampled. Taking a union bound over all $k$ largest coordinates.
We can show that with probability at least $1-1/100$, none of $k$ largest coordinates is sampled. Let $\xi$ be that event.

Second, conditioning on event $\xi$, we can show that 
\begin{align*}
\E[ \Delta_t ] \leq \frac{1}{ (100k)^{ 1/p } } \| x_{-k} \|_p.
\end{align*}

Third, applying Markov's inequality, we have
\begin{align*}
\Pr[ \Delta_t \geq a ] \leq \E[\Delta_t] /a.
\end{align*}
Choosing $a = \frac{1}{  k^{1/p} } \| x_{-k} \|_p$, we have
\begin{align*}
\Pr \left[ \Delta_t \geq \frac{1}{  k ^{1/p} } \| x_{-k} \|_p \right] \leq 1/10.
\end{align*}
\end{proof}

\begin{claim}[upper bound on $\Delta_t$]\label{cla:tail_upper_bound_on_Delta_t}
Let $\Delta_t$ be defined in Eq.~\eqref{eq:tail_def_Delta_t}. For any $C_0 \geq 1000$, we have 
\begin{align*}
\Pr_{\delta} \left[ \Delta_t < \frac{1}{  ( 2 C_0 k )^{1/p} } \| x_{- C_0 k} \|_p \right] \leq 1/10.
\end{align*}
\end{claim}
\begin{proof}

Without loss of generality, we can assume that all coordinates of $x_i$ are sorted, i.e. $x_1 \geq x_2 \geq \cdots \geq x_n$. Then we split length $n$ vector into $l$ blocks where each block has length $s= C_0 k$. Note that it is obvious $l \cdot s = n$.

For each $j \in [l]$, we use boolean variable $B_j$ to denote that if at least one coordinate in $j$-th block has been sampled. For a fixed block $j\in [l]$, the probability of sampling at least one coordinate from that block is at least
\begin{align*}
1 - \left( 1 - \frac{1}{100k} \right)^{s} = 1 - \left( 1 - \frac{1}{100k} \right)^{ C_0 k } \geq 9/10.
\end{align*}
Thus, we know $1 \geq \E[B_j] \geq 9/10$.

{\bf Warm-up.} Note the probability is not allowed to take a union over all the blocks. However, if we conditioned on that each block has been sampled at least one coordinate, then we have
\begin{align*}
\Delta_t^p = & ~ \sum_{i=1}^n \delta_{i,t} x_i^p \\
\geq & ~ \sum_{j=1}^{l - 1} x_{j s}^p \\
\geq & ~ \sum_{j=1}^{l - 1 } \frac{1}{s} \left( x_{js+1}^p + x_{js+2}^p + \cdots + x_{js+s}^p  \right) \\
= & ~ \frac{1}{s} \| x_{s} \|_p^p.
\end{align*}

{\bf Fixed.} 
For simplicity, for each $j \in [l]$, we use set $T_j$ to denote $\{ (j-1)s+1, (j-1)s+2, \cdots, (j-1)s + s\}$.

Using random walk Lemma~\ref{lem:random_walk}, with probability at least $99/100$, we have : for all $j\in \{2,\cdots, l\}$,
\begin{align*}
\sum_{j' = 2}^{j} B_{j'} > (j-1)/2.
\end{align*}
We know that with probability at least $99/100$, $B_1=1$. Then with probability at least $99/100$, we have
\begin{align*}
B_1 = 1 , \text{~and~} \sum_{j'=2}^{j} B_{j'} > (j-1)/2, \forall j \in [l].
\end{align*}
We conditioned on the above event holds. Let set $S_1 \subset [n]$ denote the set of indices $j$ such that $B_j=1$, i.e.,
\begin{align*}
S_1 = \left\{ j ~\big|~ B_j = 1, j \in [n] \backslash \{1\} \right\}.
\end{align*} 
Let set $S_0 \subset [n]$ denote the set of indices $j$ such that $B_j=0$, i.e.,
\begin{align*}
S_0 = \left\{ j ~\big|~ B_j = 0, j \in [n] \backslash \{1\} \right\}.
\end{align*}
Due to $\sum_{j'=2}^{j} B_{j'} > (j-1)/2, \forall j \in [l]$, then it is easy to see $S_1 > S_0$ and there exists a one-to-one mapping $\pi : S_0 \rightarrow S_1'$ where $S_1' \subseteq S_1$ such that for each coordinate $j \in S_0$, $\pi(j) < j$. Since we are the coordinates are being sorted already, thus
\begin{align*}
\sum_{j \in S_1} \| x_{T_j} \|_p^p = & ~ \sum_{j \in S_1'} \| x_{T_j} \|_p^p \\
= & ~ \sum_{ j \in S_1' } \| x_{T_{\pi^{-1}(j)}} \|_p^p \\
\geq & ~ \sum_{ j \in S_0 } \| x_{T_j} \|_p^p
\end{align*}
which implies that
\begin{align*}
\Delta_t^p = \sum_{i=1}^n \delta_i^p x_i^p = \sum_{j\in S_1} \| x_{T_j} \|_p^p \geq \frac{1}{2s} \| x_{-s} \|_p^p.
\end{align*}
Thus, with probability at least $9/10$, we have
\begin{align*}
\Delta_t \geq \frac{1}{ ( 2s )^{1/p} } \| x_{-s} \|_p.
\end{align*}

\end{proof}

%% file: combine.tex
\section{Putting It All Together}

Our full algorithm first applies interval forest sparse recovery algorithm in Algorithm~\ref{alg:btree} with precision parameter $\frac{\epsilon}{10}$ to obtain a set $L$ of size $O(k / \epsilon)$. 
By Theorem~\ref{thm:identification}, with probability $\frac{9}{10}$, $L$ contains a subset $T$ of size at most $k$ so that $\|x - x_T\|_2 \leq (1+\frac{\epsilon}{10}) \|x_{-k}\|_2$.
Then the algorithm feeds $L$ into the pruning procedure in Algorithm~\ref{alg:prune} also with precision parameter $\frac{\epsilon}{10}$ to get $S$.
By Theorem~\ref{thm:pruning}, with probability $\frac{9}{10}$, $|S| = O(k)$ and $\|x - x_S\|_2 \leq (1+\frac{\epsilon}{2}) \|x - x_{-k}\|_2$.
Finally, the algorithm uses set query data structure in Lemma~\ref{lem:set_query_price} to give $\hat{x}_S$ as approximation to $x_S$.
With probability at least $1 - 1/\poly(k)$, $\|\hat{x}_S - x_S\|_2^2 \leq \frac{\epsilon}{2} \|x - x_S\|_2^2$.
Therefore, $\|x-\hat{x}_S\|_2^2 = \|x-x_S\|_2^2 + \|x_S-\hat{x}_S\|_2^2 \leq (1+\frac{\epsilon}{4})(1+\frac{\epsilon}{2})^2 \|x - x_{-k}\|_2^2 \leq (1+\epsilon)^2 \|x - x_{-k}\|_2^2$.
By union bound, the overall failure probability is at most $1/4$.

\begin{algorithm}\caption{Stronger $\ell_2/\ell_2$ algorithm}
\begin{algorithmic}[1]
\Procedure{\textsc{Main}}{$x,n,k,\epsilon$} \Comment{Theorem~\ref{thm:main_sparse_recovery}}
	\State $L \leftarrow \textsc{IntervalForestSparseRecovery}(x,n,k,\frac{\epsilon}{10})$ \Comment{Algorithm~\ref{alg:btree}}
	\State $S \leftarrow \textsc{Prune}(x,n,k,\frac{\epsilon}{10},L)$ \Comment{Algorithm~\ref{alg:prune}}
	\State $\hat{x}_S \leftarrow \textsc{SetQuery}(x,n,\frac{\eps}{4},S)$ \Comment{Lemma~\ref{lem:set_query_price}}
	\State \Return $\hat{x}_S$
\EndProcedure
\end{algorithmic}
\end{algorithm}

%% file: ack.tex
\section*{Acknowledgments}
The authors would like to express their sincere gratitude to Zhengyu Wang for his great, selfless and priceless help with this project. The authors would like to thank Jaros\l{}aw B\l{}asiok, Michael Kapralov, Rasmus Kyng, Yi Li, Jelani Nelson, Eric Price, Ilya Razenshteyn, Aviad Rubinstein, David P. Woodruff and Huacheng Yu for useful discussions. The authors would like to thank anonymous reviewers for pointing out that small accuracy regime is not very interesting. Finally, the authors decided to remove it to obtain the current clean version.